\documentclass[prodmode,acmtoit]{acmsmall}

\usepackage{amsmath,epsfig,amssymb,graphicx}

\usepackage[ruled]{algorithm2e}

\SetAlFnt{\small}
\SetAlCapFnt{\small}
\SetAlCapNameFnt{\small}
\SetAlCapHSkip{0pt}
\IncMargin{-\parindent}

\newtheorem{thm}{Theorem}

\def\isp{\text{ISP}}
\def\cp{\text{CP}}
\def\ispbar{\overline{\text{ISP}}}
\def\opt{\textsf{opt}}
\def\exante{{\em ex ante }}
\def\expost{{\em ex post }}

\begin{document}

\markboth{E. Altman et al.}{Regulation of off-network pricing in a nonneutral network}

\title{Regulation of off-network pricing in a nonneutral network}
\author{Eitan Altman
\affil{INRIA}
Manjesh Kumar Hanawal
\affil{INRIA}
Rajesh Sundaresan
\affil{Indian Institute of Science}}

\begin{abstract}
Representatives of several Internet service providers (ISPs) have expressed their wish to see a substantial change in the pricing policies of the Internet. In particular, they would like to see content providers (CPs) pay for use of the network, given the large amount of resources they use. This would be in clear violation of the ``network neutrality'' principle that had characterized the development of the wireline Internet. Our first goal in this paper is to propose and study possible ways of implementing such payments and of regulating their amount. We introduce a model that includes the users' behavior, the utilities of the ISP and of the CPs, and the monetary flow that involves the content users, the ISP and CP, and in particular, the CP's revenues from advertisements. We consider various game models and study the resulting equilibria; they are all combinations of a noncooperative game (in which the ISPs and CPs determine how much they will charge the users) with a ``cooperative'' one on how the CP and the ISP share the payments. We include in our model a possible asymmetric weighting parameter (that varies between zero to one). We also study equilibria that arise when one of the CPs colludes with the ISP. We also study two dynamic game models and study the convergence of prices to the equilibrium values.
\end{abstract}



\terms{Games, Network Neutrality, Telecommunications Policy}

\keywords{off-network pricing, proportional sharing, two-sided market}

\acmformat{Eitan Altman, Manjesh Kumar Hanawal, and Rajesh Sundaresan, 2013. Regulation of off-network pricing in a nonneutral network.}

\begin{bottomstuff}
E.~Altman and R.~Sundaresan were supported by the Indo-French Centre for Applied Mathematics and by GANESH, an INRIA associates program. M.~K.~Hanawal was supported by the CONGAS European project.

Authors' addresses: E. Altman {and} M. K. Hanawal, INRIA, located at the University of Avignon, France; R. Sundaresan, Department of Electrical Communication Engineering, Indian Institute of Science, Bangalore 560012, India.
\end{bottomstuff}

\maketitle

\section{Introduction}
The initial growth of the Internet and e-commerce businesses was in the backdrop of the following ``neutrality'' principles of providing end-to-end connectivity: (1) content providers (CPs) and end users paid only the Internet service providers (ISPs) that connected them to the Internet and not any other intermediate operator, and (2) they need not know how their packets are transported in the network, but are guaranteed best effort delivery without discrimination. Indeed, \cite{HW} wrote:
\begin{quote}
``Net neutrality has no widely accepted precise definition, but usually means that broadband service providers charge consumers only once for Internet access, do not favor one content provider over another, and do not charge content providers for sending information over broadband lines to end users.''
\end{quote}
Arguably, these principles encouraged rapid innovation at the edge of the network without any interference from the network operators and made the content accessible in a nondiscriminatory fashion.

Many last mile ISPs have opposed neutrality arguing that some CPs (deriving advertising revenue from connections to customers) and applications (such as peer-to-peer or P2P streaming) used their resources without adequate compensation and that under the neutral policy the ISPs would not have an incentive to invest in network infrastructure upgrades or expansion. With a view towards encouraging investment and innovation in broadband services, the Federal Communications Commission (FCC) ruled in 2002 in favor of an {\em unregulated}, or a {\em nonneutral}, regime \cite{FCC_2002}. This decision was upheld by the courts in 2005\footnote{The background leading to this ruling and the subsequent court decision is somewhat nuanced and concerns the regulated digital subscriber line services, regulated due to historical reasons, and the unregulated cable services. The ruling in 2002 allowed unregulated cable modem services. The 2005 court decision paved the way for unregulated digital subscriber services as well. See \cite{IJoC07_TheStateofTheDebate_PehaLehirWilkie} for more details.}. This sparked a huge debate on whether the Internet should be neutral or not\footnote{Since 2005 however the FCC has been pursuing policies towards preserving a free and open Internet. The FCC enforced this neutrality principle in the matter of a network operator's interference with P2P traffic. This was overturned in a judgement \cite{Judge10_ComcastCorp}, but the courts did not disagree to FCC's support for a free and open Internet; see a subsequent statement by the FCC \cite{FCC10_FCCStatementOnComcast}. Several countries have already adopted legislation that guarantees neutrality, including Chile (the first country that adopted neutrality), the Netherlands, and Slovenia.}.

In this paper, we shall study a nonneutral regime where a CP may have to pay the last-mile ISP, a {\em side payment} in addition to the ISP that connects the CP to the internet. This is because the CPs often derive advertising revenues from their connection to the end users, a connection that is enabled by the ISP. This form of nonneutrality has been variously called the opposite of ``zero-fee'' in \cite{EconmidiesTag_2012Elsevier}, ``user discrimination'' in \cite{200903RNE_MusSchWal}, and ``off-network pricing'' in \cite{njoroge-et-al}. We shall use the terminology ``off-network pricing'' borrowed from \cite{njoroge-et-al}, and shall study mechanisms for regulating this pricing.

\begin{figure}[t]
\centering
\includegraphics[scale=.6]{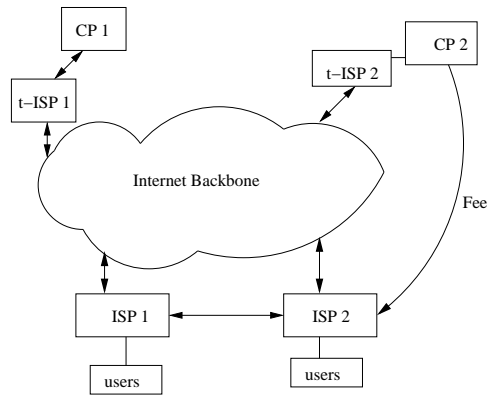}
\caption{Users-ISP-CP connections in the Internet}
\label{fig:Internet}
\end{figure}

Figure \ref{fig:Internet} shows the connection between ISPs, CPs, and end users who are consumers of content. In this paper, we shall call such end users as {\em internauts}. The internauts are connected to the Internet backbone by the last-mile ISPs. Usually, internauts do not have much choice of the ISPs -- there is either a monopolistic ISP or some times two ISPs (say ISP 1 and ISP 2). The CPs at the other end are connected to the Internet backbone via transit ISPs\footnote{These are ISPs that connect the smaller ISPs to the Internet backbone. The last-mile ISPs connect to Internet backbone through transit-ISPs. To keep the diagram simple, these are not shown in Figure \ref{fig:Internet}.}, denoted as t-ISPs in Figure \ref{fig:Internet}. The CPs usually have agreements with the t-ISPs and make payments in proportion to the bandwidth used. In the neutral regime, CP 1 only pays t-ISP 1 for connectivity to the internauts, and not to any other intermediate ISPs (ISP 1 or ISP 2). In the nonneutral off-network pricing regime a last-mile ISP (ISP 1 or ISP 2) can ask the CPs to pay for enabling connection to its internauts.

In order to focus on off-network pricing, we consider the abstracted architecture in Figure \ref{fig:cp-isp}, where there is a last-mile ISP monopoly (without ISP 2 in Figure \ref{fig:Internet}), and the combination of CP, t-ISP, and the internet backbone are combined into a single entity that is marked as CP (if there is only one CP as in Figure \ref{fig:cp-isp}). If there are several CPs, then the combination of CP $i$, t-ISP $i$, and the associated portion of the internet backbone are combined into CP $i$ (as in Section \ref{sec:multiple-CP}), with CP $i$ having a dedicated clientele\footnote{In this context, one could view internauts as applications on real end users' machines.} (internauts of class $i$). As our aim is limited to the study of regulations on off-network pricing, and because these effects are best understood when off-network pricing is studied in isolation, we do not include in our models other important considerations such as graded QoS, prioritization, investments, recurrent expenses, technology aspects, other pricing schemes such as flat-rate pricing, etc. Let us first set the stage by discussing the related and most relevant works\footnote{Analysis of nonneutral networks with QoS differentiation can be found in \cite{HermalinKatz_2007InfoEconPolicy,ChengBand_2010RandJournal,KramerWie_2009MPRA,ChoiKim_2010RandJournal}. Discussions of legal and policy implications of network neutrality regulation can be found in\cite{HW,Wu_2003JournalOfTelecom,Wu_2004JournalOfTelecom}.}.

\begin{figure}[t]
\centering
\includegraphics[width=3.49in, height=1.1in]{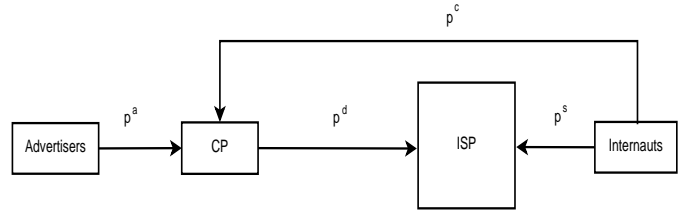}
\caption{Monetary flow in a nonneutral network.}
\label{fig:cp-isp}
\end{figure}

\subsection{Related works on off-network pricing}

\cite{EconmidiesTag_2012Elsevier} model a nonneutral network as a two-sided market\footnote{This is a market where the CPs pay the last-mile ISPs with whom they do not have a direct connection. The market is two-sided because the payment is in addition to payments made to their respective t-ISPs.}, with the CPs and a continuum of internauts connected to each other by a monopoly ISP. They show that if the ISP charges the CPs (side payments), then the ISP's profit increases, whereas the CPs' profits reduce, and there are fewer CPs that remain active at equilibrium. However, social welfare can be higher or lower than the zero-fee case depending on model parameters. Further, if a social planner is to decide the payment from the CPs to the ISP, the payment will be lower than that set by the monopoly ISP. In a similar setting, \cite{CarlsonCanon_2009WorkingPaper} studies investment incentives for ISPs and CPs, and concludes in favor of the neutral network arguing that there is higher incentive for more number of CPs and internauts to be active in this regime with greater investment and higher social welfare. \cite{njoroge-et-al} consider a duopoly ISP market and bring in several aspects such as investments by ISPs, pricing of CPs, CPs' connection decisions, consumer pricing, etc. Analyzing the resulting hierarchial 6-stage game, they conclude that in the nonneutral regime the investments will be higher with increased participation of consumer and CPs. \cite{200903RNE_MusSchWal} consider a finite number of CPs and ISPs. They conclude that social welfare is higher in a nonneutral regime if the ratio of the advertisement revenues to user price sensitivity is either high or low, and for intermediate values, a neutral regime is preferable. Closest in spirit to our work, \cite{2012xxTNSM_CouEtAl} consider two ISPs competing for internauts interested in the content of a common CP. If a regulator sets the off-networking price to maximize social utility, then everybody benefits.

The literature on the economics of off-network pricing is therefore inconclusive. There are arguments in favor of off-network pricing, or against it, or there are mixed opinions that swing one way or another depending on the model parameters. This is perhaps as one might expect in a system as large and complex as the Internet. However, if a nonneutral regime is to be considered, questions of how monetary interactions between ISPs and CPs should take place, and their influence on the internaut behavior, are worth consideration. \cite{200605JSAC_HeWal} study revenue sharing mechanisms between interconnected network operators based on a weighted proportional fairness criteria. \cite{shap1,shap2} propose the use of the Shapley value (which is known to have some fairness properties \cite{e-winter}) for deciding how revenues from internauts should be shared among the ISP and the CPs. However, these works do not consider the sizable revenues to CPs' from advertisements.

\subsection{Our Objective and Contributions}

Our objective in this paper is to consider off-network pricing, propose two ISP-CP revenue sharing mechanisms, and characterize the ensuing equilibria.

We begin with a two player game (Section \ref{sec:single-CP}) where one agent is the CP and another agent is the ISP. Both players can charge the internauts for content access on a per unit demand basis\footnote{Some consider networks with a per unit demand pricing by the ISPs as being nonneutral. But then, many big ISPs already use per unit demand pricing schemes, for example, \$10 per gigabyte scheme \cite{Book_NetworkedLife_Chiang}.}. We study the single-CP single-ISP game in two settings -- regulation of the side payment (1) before the above players set access prices (\exante {\em regulation}), and (2) after the players set access prices (\expost {\em regulation}). We then extend the results to the case when there are multiple CPs in Section \ref{sec:multiple-CP}. The demand function we consider in Section \ref{sec:single-CP} (single-CP case) is a simple, linear, decreasing function of the net price. In the multiple-CP case, demand for content from a CP is linear and decreasing in the price of that CP's content, but linear and increasing in the price of other CPs' contents, reminiscent of the Bertrand oligopoly \cite{OlogopolyPricising}. We study the equilibria for \exante regulation and \expost regulation in Sections \ref{subsec:bargain-before-multicp} and \ref{subsec:bargain-after-multicip}. In Section \ref{sec:ExclusiveContract}, we study the impact of a CP having an exclusive contract with the ISP. The paper concludes with a discussion in Section \ref{sec:discussion}. The paper comes with a fairly large appendix. It includes (1) a discussion of an appropriate model for the demand function in multiple CP settings when some flows may drop out thereby freeing ISP capacity for the remaining flows (Appendices \ref{app:GeneralDemandFunction}, \ref{app:proof-thm:multi-cp-bb-zerodemand}, and \ref{app:proof-thm:multicp-bb-mixed}), (2) proofs of main results for the multiple CP case all of which are quite elementary but at times tedious, and (3) two dynamic models of the game studied in Section \ref{sec:multiple-CP} and their convergence analysis (Section \ref{app:Dynamics}).

Our models are different from other proposed revenue sharing schemes in \cite{200903RNE_MusSchWal,200605JSAC_HeWal,shap2}. Following \cite{claudia}, our regulation schemes attempt to share revenue according to a proportional sharing paradigm. Though simple and stylized, our models are tractable and capable of providing some interesting insights to policy makers. This is our biggest motivation to publish this work. We invite the reader to look at Section \ref{sec:single-CP}, after the detailed description of the model, for a list of implications.

\section{The Case of a Single CP and a Single ISP}
\label{sec:single-CP}

We first begin with the simple case of a single CP and a single ISP. All the internauts are connected to the ISP, and can access the content of the CP only through the ISP. See Figure \ref{fig:cp-isp} for a payment flow diagram. The various parameters of the off-network pricing game are as follows.

\begin{center}
\begin{tabular}{c|p{4.4in}}
\hline \hline
Parameter & Description \\ \hline \hline
$p^s$        & Price per unit demand paid by the internauts to the ISP. This can be positive or negative, and when negative, ISP pays the CP. \\
$p^c$        & Price per unit demand paid by the users to the CP. This too can be positive or negative. \\
$d(p^s,p^c)$ & Demand as a function of prices. We shall take this to be $d(p^s,p^c) = [D_0 - \alpha(p^s+p^c)]_+$, where $[x]_+ = \max\{x,0\}$ is the positive part of $x$. \\
$p^a$        & Advertising revenue per unit demand, earned by the CP. This satisfies $p^a \geq 0$. \\
$p^d$        & Price per unit demand paid by the CP to the ISP. This can be either positive or negative. \\
$U_{\isp}$   & The revenue or utility of the ISP, given by $d(p^s, p^c) (p^s + p^d)$. \\
$U_{\cp}$    & The revenue or utility of the CP, given by $d(p^s, p^c) (p^c + p^a - p^d)$. \\
$\gamma$      & Relative weight of the ISP with respect to the CP; $0 < \gamma < 1$. \\
\hline \hline
\end{tabular}
\end{center}
\cite{nn} noted that if $p^d$ is controlled by either of the players and is set jointly with that player's access price, then the price competition between the ISP and the CP results in zero demand at equilibrium, which is not favorable to any of the agents. This motivates us to study the case when $p^d$ is set by a neutral third party whom we refer to as `regulator'. The regulator can be a law enforcing agency which decides the side payment taking into account the market powers of the players as described below. We consider two interesting games.

The timing for the first game, under \exante regulation, is as follows.

\begin{enumerate}
  \item The regulator sets the payment $p^d$ from the CP to the ISP.
  \item The CP sets the price $p^c$. Simultaneoulsy, the ISP sets the price $p^s$.
  \item The internauts react to the prices and set the demand\footnote{The ISP and the CP set prices for their roles in the service rendered to the internauts. The resulting demand for content depends on the joint prices only through the sum, which is the total price per unit demand seen by the internauts.} $d(p^s,p^c) = [D_0 - \alpha(p^s+p^c)]_+$.
\end{enumerate}

In the second game, under \expost regulation, the timing is as follows.

\begin{enumerate}
  \item The CP and the ISP set their respective access prices $p^c$ and $p^s$ simultaneously.
  \item The internauts react to the prices and set the demand.
  \item The regulator sets the payment $p^d$ from the CP to the ISP.
\end{enumerate}

The first game arises when the charges per unit demand can change over a comparatively faster time-scale while the CP-ISP price $p^d$ changes over a slower time-scale. The second one is an interesting case when the prices per unit demand charged to the internauts varies over a slower time-scale, but the CP-ISP price changes over a faster time-scale. We analyze both models via backward induction and identify the equilibria.

For a fixed $p^s$ and $p^c$, the mechanism used by the regulator to decide payment $p^d$ from the CP to the ISP is as follows:
\[
  p^{d*} \in \arg \max_{p^d} U_{\isp}^{\gamma} \times U_{\cp}^{1 - \gamma}.
\]
The parameter $\gamma$ relates the market power of the ISP to that of the CP\footnote{The Spanish ISP ``Telefonica'' announced on 8 February 2010 that it considered charging Google. This is an indication that the bargaining power of Google was weaker than that of Telefonica in the Spanish telecommunications market in 2010. If $p^d$ were to be set via a binding agreement between the CP and the ISP, and the regulator was the enforcing agent, $\gamma$ may indicate the relative strength of the ISP with respect to the CP. If $\gamma = 1/2$, then the maximization is equivalent to that of the product of the utilities of the ISP and the CP. This is then standard Nash bargaining outcome \cite{nash} for resource allocation, known in networking as the proportional fair allocation \cite{kelly}. Ultimately, under \expost regulation, we shall see that ISP gets $\gamma$ fraction of the total revenue, thereby justifying its interpretation as market power. The quantity $\gamma$ can also be interpreted as measuring the ``degree of cooperation''; see \cite{amar}.}. Note that in both the games only $p^d$ is set according to the above regulation mechanisms, while the other prices are set simultaneously and are strategic actions.

In the \exante regulation game, the regulator sets $p^d$ knowing that $p^s$ and $p^c$ will be chosen subsequently by the players; $U_{\isp}$ and $U_{\cp}$ are then taken to be the {\em equilibrium utilities} at side payment level $p^d$ (with an appropriate selection if there are many equilibria). In the \expost regulation game, the players choose $p^s$ and $p^c$ knowing how the regulator will set $p^d$ subsequently. The following is a motivating list of policy implications stemming from the results for the single-CP single-ISP games.

\begin{enumerate}
\item  In both cases, there exists a pure strategy Nash equilibrium, in a sense that will be made precise, with strictly positive demand and strictly positive utilities for the agents.
    While it is possible that demand can be zero in the \exante regulation game, it is always positive in the \expost regulation game with unique equilibrium prices. For a policy maker, such a conclusion is very useful -- \expost regulation never produces a stalemate zero-demand outcome where no player has a positive revenue; \exante regulation may\footnote{However, see the multiple CP case.}.
\item In all cases with strictly positive demand, the internauts pay the ISP. But the internauts pay the CP only if the advertising revenue is small. Otherwise the CP subsidizes the internauts. This is natural, of course, yet pleasing to see that the model bears this out.
\item If either of the agents has control over $p^d$ and sets it jointly with its access price, the equilibrium demand is zero \cite{nn}. None of the parties benefit from this situation. On the contrary, if $p^d$ is under the control of a disinterested regulator, with timing as in either the \exante or the \expost regulation game, there is an equilibrium where every one benefits. This is the key insight gained from our analysis, that some sort of regulation can bring benefits to all.
\item Interestingly, if the regulator applies \exante regulation and the strictly positive demand equilibrium ensues, the payments by the internauts and resulting utilities of all agents are independent of the actual value of $p^d$ and $\gamma$. Any $p^d$ paid by the CP is collected from the internauts and is further returned back to the internauts by the ISP. This lack of sensitivity of the outcome to the value of $p^d$ is a robustness property that can be quite reassuring to a policy maker who may worry about market distortions arising from the regulation. It only matters that there is some regulation, the actual value of the regulated price $p^d$ is irrelevant.
\item For the \exante regulation game, the demand settles at a lower value and the internauts pay a higher price per unit demand, when compared with the \expost regulation game which results in greater welfare (if welfare is determined by the price that the internauts pay). Again, this is an observation of value to the policy maker that is interested in maximizing internauts' welfare.
\item For the \exante regulation game, both the players end up with equal revenues. This attributes equal importance to both players. For the \expost regulation game, they share the net revenue in the proportion of their relative weights.
\item If $\gamma \in \left[\frac{4}{9}, \frac{5}{9}\right]$, then both agents prefer \expost regulation. For $\gamma > 5/9$, the ISP prefers \expost regulation, and for $\gamma < 4/9$, the CP prefers \expost regulation.
\item Finally, in view of the fourth point above, one recovers the neutral regime in our \exante regulation game by setting $p^d = 0$. The internauts are thus indifferent to the neutral regime and the nonneutral off-network regime under \exante regulation.
\end{enumerate}

The choice of \exante regulation game or \expost regulation game in this single-CP single-ISP depends on societal priorities. {\em Ex ante} regulation is robust, makes internauts insensitive to nonneutral versus neutral regimes, but has the possibility of a stalemate zero-demand equilibrium or an equilibrium with lower welfare for internauts\footnote{We will later see in Section \ref{sec:multiple-CP} that for multiple CPs, the \expost regulation game may have no pure strategy Nash equilibria for some parameter ranges.}. Our goal in this paper is not to discuss societal priorities, but merely to provide conclusions as above that will aid a policy maker in his decision making.

With these motivating remarks, we shall now proceed to state these claims in a precise fashion and to prove them. In subsequent sections we shall study the extension of the above results to the case of multiple CPs and to the case of an exclusive contract between one of the CPs and the ISP.

\subsection{Ex ante regulation}
\label{subsec:bargain-before}

We first consider the case where the regulator sets $p^d$, knowing that the players will subsequently play a simultaneous action game where the ISP and CP will choose $p^s$ and $p^c$, respectively. Our main result here is summarized as follows.

\begin{thm}
\label{thm:1cp-bargainbefore}
In the \exante regulation game, we have the following complete characterization of all pure strategy Nash equilibria.

(a) Among profiles with strictly positive demand, there is a unique pure strategy Nash equilibrium with the following properties:
  \begin{itemize}
    \item The uniqueness is up to a free choice of $p^d$.
    \item At equilibrium, we have:
      \begin{eqnarray}
        \label{eqn:p1}
        p^s & = & \frac{D_0 + \alpha p^a}{3 \alpha} - p^d, \\
        \label{eqn:p2}
        p^c & = & \frac{D_0 - 2\alpha p^a}{3 \alpha} + p^d.
      \end{eqnarray}
    \item The net internaut payment per demand $p^s + p^c$ is unique and is given by
      $$p^s + p^c = \frac{2D_0 - \alpha p^a}{3 \alpha}.$$
    Any $p^d$ paid by the CP is collected from the internauts and further returned back to the internauts by the ISP.
    \item The demand is unique and is given by $(D_0 + \alpha p^a)/3 > 0$.
    \item The utilities of the ISP and CP are equal and given by
        $$U_{\isp} = U_{\cp} = \frac{(D_0 + \alpha p^a)^2}{9\alpha}.$$
  \end{itemize}

(b) For each choice of $p^d$, a strategy profile $(p^s, p^c)$ constitutes a Nash equilibrium with zero demand if and only if the following two inequalities hold:
  \begin{eqnarray}
    \label{eqn:zero-demand-eq1}
    p^s & \geq & D_0 / \alpha + p^a - p^d, \\
    \label{eqn:zero-demand-eq2}
    p^c & \geq & D_0 / \alpha + p^d.
  \end{eqnarray}
\end{thm}

\begin{proof}
We first observe that at equilibrium, $U_{\isp}$ and $U_{\cp}$ are both nonnegative. If not, the ISP (resp. CP) has strictly negative utility. He can raise the price $p^s$ (resp. $p^c$) to a sufficiently high value so that demand becomes zero, and therefore $U_{\isp} = 0$ (resp. $U_{\cp} = 0$). Thus a deviation yields a strict increase in utility and therefore cannot be an equilibrium. It follows that at equilibrium, we may take the revenues per demand for the ISP and CP to be nonnegative, i.e., $p^s + p^d \geq 0$, and $p^c + p^a - p^d \geq 0$.

We next deduce (b), which is a characterization of all the pure strategy NE with zero demand. Consider a fixed $p^d$. If a pair $(p^s, p^c)$ were an equilibrium with zero demand, then clearly
$$D_0 \leq \alpha(p^s + p^c),$$
and
$$U_{\isp} = d(p^s, p^c) \times (p^s + p^d) = 0.$$
Moreover, the ISP should not be able to make his utility positive, i.e., any $p^s$ that makes demand strictly positive, $p^s < D_0 / \alpha - p^c$, must also render price per unit demand zero or negative, $p^s + p^d \leq 0$. This can happen only if $(D_0 / \alpha - p^c) + p^d \leq 0$ which is the same as (\ref{eqn:zero-demand-eq2}). Similarly, the CP should not be able to make his utility positive, i.e., any $p^c$ that makes demand strictly positive, $p^c < D_0 / \alpha - p^s$, must render CP price per unit demand nonpositive, $p^c + p^a - p^d \leq 0$. This can happen only if $(D_0 / \alpha - p^s) + p^a - p^d \leq 0$ which is the same as (\ref{eqn:zero-demand-eq1}). This proves the necessity of (\ref{eqn:zero-demand-eq1}) and (\ref{eqn:zero-demand-eq2}). We now show sufficiency. Let (\ref{eqn:zero-demand-eq1}) and (\ref{eqn:zero-demand-eq2}) hold. Then addition of $p^c$ to both sides of (\ref{eqn:zero-demand-eq1}) and some rearrangement yields
\begin{equation}
  \label{eqn:zero-demand-proof}
  p^c + p^a - p^d \leq p^s + p^c - D_0/\alpha.
\end{equation}
Since the left side is the revenue per unit demand for the CP, it must be nonnegative, and hence $p^s + p^c - D_0/\alpha \geq 0$, which upon rearrangement yields $D_0 - \alpha(p^s + p^c) \leq 0$. The demand $d(p^s, p^c)$ is therefore zero. Let us now consider a deviation by the CP for a fixed ISP price $p^s$ that satisfies (\ref{eqn:zero-demand-eq1}). We will show that the least deviation (decrease in price) that sets the demand at the threshold of positivity results in a negative revenue per demand for the CP. Indeed, this critical price $q^c$ that sets the demand at the threshold of positivity satisfies the equation
\[
  D_0 - \alpha(q^c + p^s) = 0.
\]
Again, addition of $q^c$ to both sides of (\ref{eqn:zero-demand-eq1}) yields, by the same steps above that led to (\ref{eqn:zero-demand-proof}),
\[
  q^c + p^a - p^d \leq p^s + q^c - D_0/\alpha = 0.
\]
Further reduction in price to make demand strictly positive only results in negative revenue and negative utility. Consequently, the CP does not have a deviation that yields a higher revenue. A similar argument shows that, under (\ref{eqn:zero-demand-eq2}), the ISP can make demand strictly positive only if its revenue is negative. It too does not have a deviation with a strictly greater utility. Thus (\ref{eqn:zero-demand-eq1}) and (\ref{eqn:zero-demand-eq2}) constitute zero demand equilibrium prices.

Let us now search for an equilibrium with a strictly positive demand. Such a $(p^s,p^c)$ must lie in the interior of the set of all pairs satisfying $D_0 \geq \alpha(p^s + p^c)$. As $U_{\isp}$ is concave in $p^s$ for a fixed $p^c$ and $p^d$, and $U_{\cp}$ is concave in $p^c$ for a fixed $p^s$ and $p^d$, the equilibrium point must satisfy the following first order optimality conditions
\begin{eqnarray*}
  \frac{\partial U_{\isp}}{\partial p^s} =  \frac{\partial}{\partial p^s}(D_0 - \alpha(p^s + p^c))(p^s + p^d)=0 \\
\end{eqnarray*}
\begin{eqnarray*}
  \frac{\partial U_{\cp}}{\partial p^s} & = & \frac{\partial}{\partial p^s}(D_0 - \alpha(p^s + p^c))(p^c + p^a - p^d)=0. \\
\end{eqnarray*}
Solving these two simultaneous equations in the variables $p^s$ and $p^c$, we see that $p^s$ and $p^c$ are given by (\ref{eqn:p1}) and (\ref{eqn:p2}), respectively. Note that $p^d$ is free parameter. Once this is chosen, the choice fixes both $p^s$ and $p^c$. This proves the second item. We shall return to prove the first item after proving the others.

Adding (\ref{eqn:p1}) and (\ref{eqn:p2}), we see that $p^s + p^c$ is a constant for any such equilibrium. Choice of $p^d$ fixes both $p^s$ and $p^c$. This is true for any Nash equilibrium with a strictly positive demand. Furthermore, any $p^d$ that is paid reduces $p^s$ by that amount and increases $p^c$ by the same amount. This proves the third item.

The last two items follow by direct substitutions into $d(p^s, p^c)$, $U_{\isp}$, and $U_{\cp}$.

As a consequence of the observation that $U_{\isp} = U_{\cp}$ at any equilibrium regardless of the value of $p^d$, we have
\[
  U_{\isp}^{\gamma} \times U_{\cp}^{1 - \gamma}
\]
is independent of $p^d$ at any equilibrium, for any fixed relative weight $\gamma \in (0,1)$. The regulator may thus pick any $p^d$. This proves the first item. (This observation holds even for zero-demand equilibria). The proof is now complete.
\end{proof}

{\em Remarks}
1) Every choice of $p^d$ can also result in the undesirable zero-demand equilibria, and not just the desirable equilibrium with strictly positive demand.

2) For this strictly positive demand equilibrium, the natural choices of $p^d$ are those that make $p^d = 0$, i.e., there is no payment from CP to ISP, or $p^c = 0$, there is no payment from the internauts to the CP, or $p^s = 0$, there is no payment from the user to the ISP.

3) If one places the additional restriction that $p^s \geq 0$, the only effect of this constraint is that the choice of $p^d$ is restricted to $p^d \leq (D_0 + \alpha p^a)/(3 \alpha)$, and the above theorem continues to hold.

\subsection{Ex post regulation}
\label{subsec:bargain-after}

We next consider the case when the CP and ISP decide on their respective prices first, knowing that the regulator will set the side payment subsequently.

\begin{thm}
\label{thm:1cp-bargainafter}
In the \expost regulation game, there is a unique pure strategy Nash equilibrium with the following properties:
\begin{itemize}
  \item The uniqueness is up to a free choice of either $p^s$ or $p^c$. Without loss of generality, we may assume a free $p^s$.
  \item At equilibrium, the net user payment per demand is uniquely given by
    \begin{eqnarray*}
      p^s + p^c = \frac{D_0 - \alpha p^a}{2 \alpha}.
    \end{eqnarray*}
  \item The demand is unique and is given by $(D_0 + \alpha p^a)/2 > 0$.
  \item The regulator will set $p^d$ so that the net revenue per demand $p^s + p^c + p^a = \frac{D_0 + \alpha p^a}{2 \alpha}$ is shared in the proportion $\gamma$ and $1 - \gamma$ by the ISP and the CP, respectively.
\end{itemize}
\end{thm}

\begin{proof}
As in the previous section, it is clear that the revenues per demand and the utilities for both agents are nonnegative. If this is not the case, the aggrieved CP or the ISP guarantees himself a strictly larger zero utility by raising the price under his control so that demand reduces to 0.

Let us now perform a search for equilibria with strictly positive demand. Such a $(p^s,p^c)$ is an interior point among all those pairs that satisfy $D_0 - \alpha(p^s + p^c) \geq 0$. Consider a fixed interior point $(p^s, p^s)$. The regulator sets $p^d$ to
\begin{eqnarray*}
  \arg \max_{p^d} U_{\isp}^{\gamma} \times U_{\cp}^{1 - \gamma}
    =  \arg \max_{p^d} \left[ \gamma \log (p^s + p^d) + (1 - \gamma) \log (p^c + p^a - p^d) \right],
\end{eqnarray*}
where the equality follows because the demand can be pulled out of the optimization. The optimization is over the set of $p^d$ that ensure that the arguments inside the logarithm remain strictly positive. It is easy to see that the latter function is concave in $p^d$, and thus the maximizing $p^d$ satisfies
\[
  \frac{\gamma}{p^s + p^d} - \frac{1 - \gamma}{p^c + p^a - p^d} = 0,
\]
which yields $p^d = \gamma(p^c + p^a) - (1-\gamma)p^s$.

Substitution of this $p^d$ yields
\begin{eqnarray*}
  p^s + p^d & = & \gamma (p^s + p^c + p^a) \\
  p^c + p^a - p^d & = & (1 - \gamma) (p^s + p^c + p^a).
\end{eqnarray*}
Clearly, $p^s + p^c + p^a$ is the net revenue per demand for both ISP and CP put together, and the ISP and the CP share this booty in the fraction of their relative weights.

Knowing this action of the regulator, the ISP responds to the CP's $p^c$ by maximizing
\[
  U_{\isp} = d(p^s,p^c) (p^s + p^d) = (D_0 - \alpha(p^s + p^c) \times \gamma(p^s + p^c + p^a).
\]
This is a concave function of $p^s$, and the maximum is at
\begin{equation}
  \label{eqn:BR-isp}
  p^s = \frac{D_0 - \alpha p^a}{2 \alpha} - p^c.
\end{equation}

Similarly, for an ISP's $p^s$, the CP's best response is
\[
  p^c = \frac{D_0 - \alpha p^a}{2 \alpha} - p^s,
\]
which is the same equation as (\ref{eqn:BR-isp}).

At equilibrium, we thus have $p^s+p^c$ uniquely determined and given by the second item. A substitution yields that the demand is given by
\[
  d(p^s, p^c) = D_0 - \alpha(p^s + p^c) = \frac{D_0 + \alpha p^a}{2},
\]
which proves the third item.

The revenue per demand is easily seen to be $(D_0 + \alpha p^a)/(2 \alpha)$. Further substitution yields that net revenue is
$d(p^s, p^c) (p^s + p^c + p^a) = (D_0 + \alpha p^a)^2 / (4 \alpha),$
a strictly positive quantity shared in proportion of their relative weights by the ISP and CP. This proves the last item.

Finally, for any $p^s$, the regulator will set $p^d$ to ensure this proportion, and thus $p^s$ may be taken as a free variable. Each $p^s$ and $p^c$ satisfying the above conditions is a Nash equilibrium. This proves the first item.

Finally, it still remains to prove that there is no zero-demand equilibrium. Suppose that $(p^s, p^c)$ is such that we get a zero-demand, i.e., $D_0 \leq \alpha(p^s + p^c)$. With $\varepsilon = (D_0 + \alpha p^a)/2 > 0$, the ISP can set his new price to
$$q^s = D_0 / \alpha - p^c - \varepsilon/\alpha$$
yielding a demand $D_0 - \alpha(q^s + p^c) = \varepsilon > 0$ and a revenue
$$\gamma (q^s + p^c + p^a) = \gamma (D_0 / \alpha - \varepsilon / \alpha + p^a) = \gamma \varepsilon/\alpha > 0,$$
and therefore a strictly positive utility. A unilateral deviation yields the ISP a strict increase in his utility. Thus a $(p^s, p^c)$ with zero demand cannot be a pure-strategy equilibrium. This concludes the proof.
\end{proof}

{\em Remarks}: 1) The equilibrium utility for the ISP under \expost regulation is easily seen to be $9 \gamma/4$ fraction of that under \exante regulation. Clearly then, the ISP prefers \expost regulation if $\gamma \geq 4/9$.

2) Similarly, the equilibrium utility for the CP under \expost regulation is $9(1 - \gamma)/4$ fraction of that under \exante regulation. The CP prefers \expost regulation if $1 - \gamma \geq 4/9$ or $\gamma \leq 5/9$.

3) Thus, if $\gamma \in \left[ \frac{4}{9},  \frac{5}{9} \right]$, both prefer \expost regulation. For $\gamma > 5/9$, ISP prefers \expost regulation while CP prefers \exante regulation. Opposite is the case when $\gamma < 4/9$.

\section{The case of multiple content providers}
\label{sec:multiple-CP}

We now consider the case when there are several CPs. Internauts connect to each of the CPs through the single ISP. See Figure \ref{fig:multiplecp-isp}. The parameters of this game are given in the following table.

\begin{center}
\begin{tabular}{c|p{4.4in}}
\hline \hline
Parameter & Description \\ \hline \hline
$n$            & Number of CPs. \\
$p^s_i$        & Price per unit demand paid by the users to the ISP for connection to CP $i$. This can be positive or negative. \\
$p^c_i$        & Price per unit demand paid by the users to CP $i$. This too can be positive or negative. \\
$p^a_i$        & Advertising revenue per unit demand, earned by the CP. This satisfies $p^a_i \geq 0$. \\
$p^d_i$        & Price per unit demand paid by the CP to the ISP. This can be either positive or negative. \\
$p^x$          & Vectors of aforementioned prices, where $x$ is one of $s,~c,~a,~d$.\\
$d_i(p^s,p^c)$ & Demand for CP $i$ as a function of the prices. See (\ref{eqn:demand-multi-CP}) below and the following discussion. \\
$r_{\cp, i}$   & The revenue per unit demand of CP $i$, given by $p^c_i + p^a_i - p^d_i$.\\
$r_{\isp, i}$  & The revenue per unit demand of ISP coming from content provided by CP $i$, given by $p^s_i + p^d_i$. \\
$U_{\isp}$     & The revenue or utility of the ISP, given by $\sum_i d_i(p^s, p^c) (p^s_i + p^d_i)$. \\
$U_{\cp, i}$   & The revenue or utility of the CP, given by $d_i(p^s, p^c) (p^c_i + p^a_i - p^d_i)$. \\
$\gamma_i$     & Relative weight of the ISP with respect to the CP. \\
\hline \hline
\end{tabular}
\end{center}

\subsection{Demand function: Strictly positive demands}
\label{subsec:demand-function-strict-positive}

The demand function for content from CP $i$ is such that it depends on $p^s$ and $p^c$ only through the sum $p^s + p^c$, the vector of net payment per unit demand from the internauts. An interesting feature we wish to model is a {\em positive correlation in demand with respect to others' prices}. Assume that the ISP has a fixed capacity/bandwidth of $W$. If CP $i$ and ISP increase their prices for content from CP $i$, demand for this content naturally goes down. On the other hand, when the price for CP $j$ content increases, where $j \neq i$, the decrease in demand for content from CP $j$ frees up some capacity. This provides a marginally better delay experience for the internauts of other CPs, and particularly internauts of CP $i$. This positive effect creates a marginal increase in the demand for content from the other CPs, and in particular, an increase in the demand for content from CP $i$. We model this correlation effect by setting the demand functions to be
\begin{equation}
  \label{eqn:demand-multi-CP}
  d_i(p^s, p^c) = \left[ D_0 - \alpha (p^s_i + p^c_i) + \beta \sum_{j: j \neq i} (p^s_j + p^c_j) \right]
\end{equation}
provided each of the demands are strictly positive. Here $\beta$ is the sensitivity parameter for the increase in demand for CP $i$ content per unit increase in price of CP $j$ content, when $j \neq i$.

While (\ref{eqn:demand-multi-CP}) is justifiable when all demands are positive, further thought suggests that it must be refined a little to account for the following. When the price $p^s_i+p^c_i$ charged to CP $i$ internauts is such that it forces demand $d_i$ to be zero, then any additional increase in $p^s_i + p^c_i$ simply continues to hold this demand $d_i$ at zero. The capacity freeing and the consequent phenomenon of increase in demand for other CPs' contents no longer occurs, and additional price rise for CP $i$ content will have no further tangible effect on other internauts' behavior. We shall return to this refinement in the next subsection after addressing some points on the positive demand case.

\begin{figure}[t]
\centering
\includegraphics[width=3.45in, height=2in]{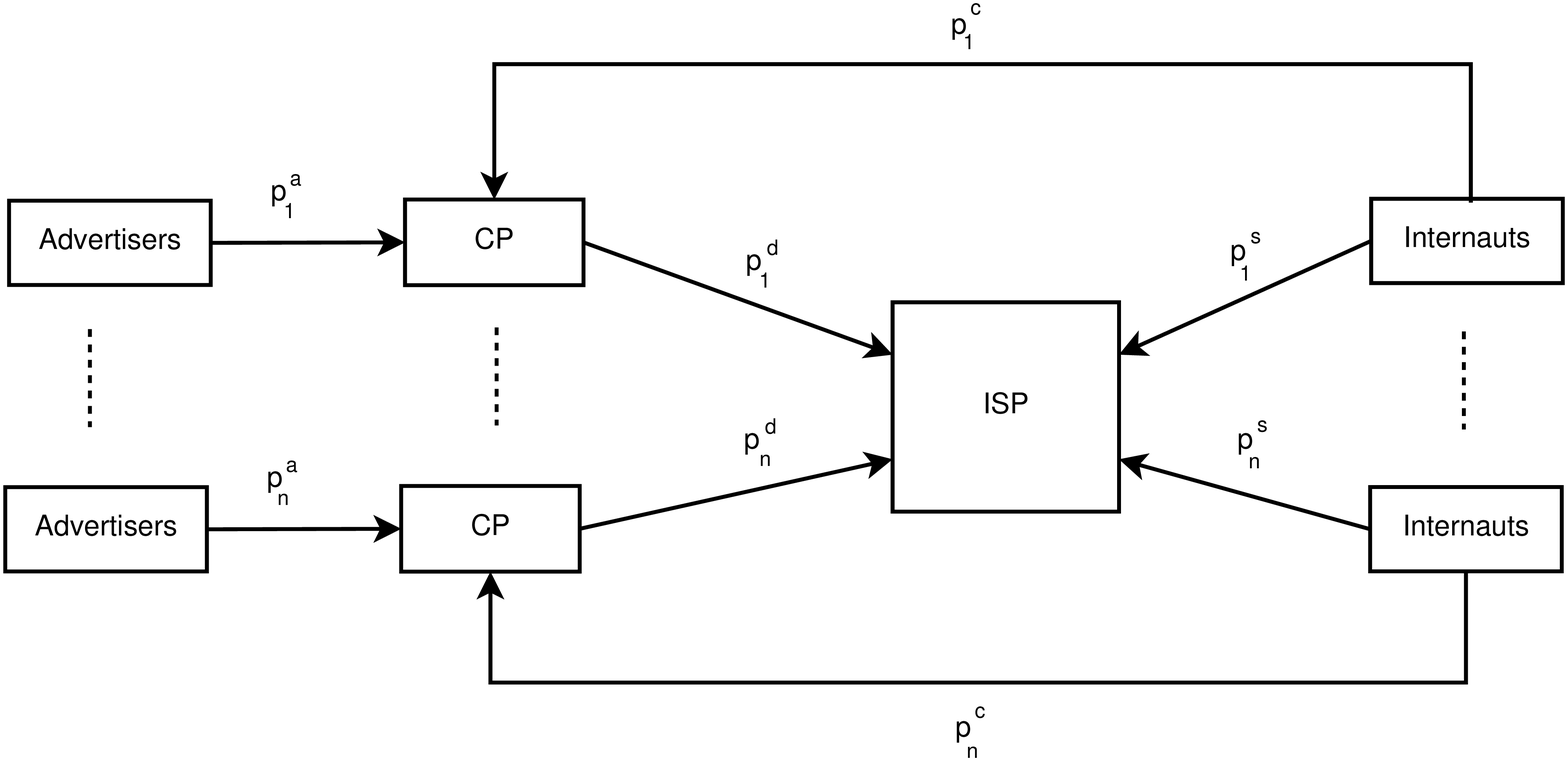}
\caption{Monetary flow in a nonneutral network with multiple CPs.}
\label{fig:multiplecp-isp}
\end{figure}

Let the evaluations in (\ref{eqn:demand-multi-CP}) be strictly positive for each $i$. If this is placed as a requirement, one could view it as a joint constraint on the actions of the ISP and CPs: given the other prices, CP $i$ will not set too high a price that makes $d_i$ zero; neither will the ISP. We may write $d_i(p^s, p^c) > 0$ for every $i$ as
\begin{equation}
  \label{eqn:price-constraint-ub}
  D_0 - \alpha (p^s_i + p^c_i) + \beta \sum_{j: j \neq i} (p^s_j + p^c_j) > 0, ~ i = 1, \ldots, n,
\end{equation}
which is compactly summarized as follows. Define the matrix $A_n = (\alpha + \beta) I_n - \beta J_n$ where $I_n$ is the identity matrix of size $n \times n$, and $J_n$ is the square matrix with all-one entries of size $n \times n$. The matrix $A_n$ has diagonal entries $\alpha$ and all off-diagonal entries $-\beta$. Also define $E_n$ to be the all-one vector of size $n \times 1$. Then the constraint (\ref{eqn:price-constraint-ub}) in matrix notation is
\begin{equation}
  \label{eqn:price-constraint-ub-alt}
  D_0 E_n - A_n (p^s + p^c) > 0.
\end{equation}
Sum up the components in (\ref{eqn:price-constraint-ub}) over all $i$ and set the sum price $P = \sum_{i} (p^s_i + p^c_i)$, and we see that the total demand is
\[
  nD_0 - (\alpha - (n-1)\beta)P
\]
under the assumption that each $d_i$ is strictly positive. For this total demand to be negatively correlated with the average price per unit demand $P/n$, we must have that
\begin{equation}
  \label{eqn:negative-correlation}
  (n-1)\beta \leq \alpha,
\end{equation}
an assumption that we make from now on\footnote{This condition also arises from assumption (D) in \cite{Informs2004_BernsFeder}.}. As before we assume that $p^s_i$ and $p^c_i$ can be negative, i.e., the ISP and CP can pay the internauts for their usage, with a consequent increase in demand.

Under the constraint (\ref{eqn:price-constraint-ub}), $U_{\isp}$ given by
\[
  U_{\isp} = \sum_{i=1}^n d_i(p^s,p^c) (p^s_i + p^d_i)
\]
is a concave quadratic function\footnote{Simple calculations show that the Hessian matrix is $-2A_n$. To see that it is negative semidefinite, observe that $-2A_n  = - 2 \alpha \times [(1 - \rho) I_n + \rho J_n]$ where $\rho = -\beta/\alpha$. The matrix $(1 - \rho) I_n + \rho J_n$ has $1-\rho$ as an eigenvalue repeated $n-1$ times and $1 + \rho (n-1)$ once, and is therefore positive semidefinite by our assumption (\ref{eqn:negative-correlation}). (It is positive definite if there is strict inequality in (\ref{eqn:negative-correlation})). Consequently, the Hessian $-2A_n$ is negative semidefinite, and $U_{\isp}$ is a concave function of $p^s$.} of the vector of ISP prices $p^s$.

\subsection{General demand function}
\label{subsec:GeneralDemandFunction}

As alluded to above, the demands in (\ref{eqn:demand-multi-CP}) have to be refined to account for the lack of further positive correlation after a demand reaches zero. See the discussion in the paragraph following the one containing (\ref{eqn:demand-multi-CP}). We present the detailed derivation of the general demand function in Appendix \ref{app:GeneralDemandFunction}. For a given price vector $p=(p_1,p_2,\cdots, p_n)$, where $p_i=p^s_i+p^c_i$, the demand for content $i$ is defined as follows:
\begin{equation}
  \label{eqn:GeneralDemand}
  d_i(p) = \left\{
    \begin{array}{ll}
      D_0 - \alpha p_i + \beta \left( \sum_{j < k^*+1, j \neq i} p_j \right) + (n-k^*) \beta T(k^* + 1), & i = 1, \ldots, k^* \\
      0, & i > k^*
    \end{array}
    \right.
\end{equation}
where
\[
T(k) := \frac{D_0 + \beta \sum_{j: j < k} p_j}{\alpha - (n-k) \beta}, \quad k = 1, 2, \ldots, n,
\]
and $k^*$ is the smallest index among $k = 0, 1, \ldots, n$ for which
\begin{eqnarray*}
  p_i & < & T(i), \quad i = 1, \ldots, k \nonumber \\
  p_{k+1} & \geq & T(k+1).
\end{eqnarray*}

In the study of problems with linear demand functions of the form (\ref{eqn:demand-multi-CP}), the analysis is usually restricted to the price vectors that results in positive demand from each player. To the best of our knowledge, the above demand function that completely characterizes the demands for any price vector is new. The considerations that lead to this demand function are given in Appendix \ref{app:GeneralDemandFunction}. This generalization is an another modest contribution of this paper.

To get more insight into the general demand function we summarize the demands for the case of two CPs. In Figure \ref{fig:Demand Region} we plot all the possible demands as a  function of internaut prices ($p = (p_1,p_2)$). As shown, we can divide the set of prices into four regions. A description of the regions is given below.

\begin{figure}[t]
  \centering
  \includegraphics[width=3.25in, height=2.6in]{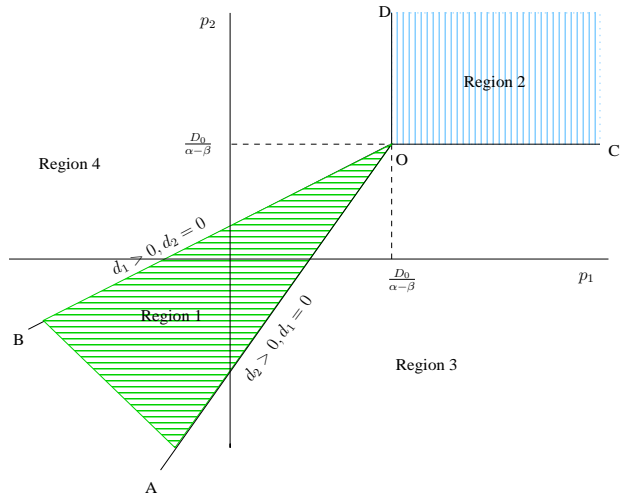}
  \caption{Characterization of the demand region. (Scales of abscissa and ordinate are different).}\label{fig:Demand Region}
\end{figure}

\begin{enumerate}
\item Denote the vector of net prices by $p=(p_1,p_2)$. If it lies in the interior of the region bounded by lines $AO$ and $BO$, denoted as Region $1$, demands for contents from both the CPs are strictly positive.
\item In the rectangular region enclosed between lines $OC$ and $OD$, denoted as Region $2$, the demands for contents from both CPs are zero.
\item In the region enclosed between the lines $AO$ and $OC$, denoted as Region $3$, demand of CP 1 content is zero and that of CP $2$ content is positive. Any point $p$ that lies on the line $AO$ is such that $p_1 = (D_0 + \beta p_2)/\alpha$ with $p_2 < D_0/ (\alpha -\beta)$.
\item In the region enclosed between the lines $BO$ and $OD$, denoted as Region $4$, demand of CP $2$ content is zero and that of CP $1$ is positive. Any point $p$ that lies on line $BO$ is such that $p_2 = (D_0 + \beta p_1)/\alpha$ with $p_1 < D_0/ (\alpha -\beta)$.
\end{enumerate}

\subsection{Timing of actions}
\label{subsec:timing-of-actions-multiCP}

The timing of actions for the games are as follows. For the \exante regulation game:

\begin{enumerate}
  \item The regulator sets the side payment from each CP to the ISP, separately and simultaneously. This can be positive or negative. In deciding the amount paid by CP $i$, the regulator shall bring only that revenue into consideration which is generated by internauts connected to CP $i$.
  \item All the CPs choose their price $p^c_i$. The ISP chooses the vector $p^s$. All these actions are taken simultaneously.
  \item The internauts react to the prices and set their demands as per the discussion in the previous subsection.
\end{enumerate}

For the \expost regulation game:
\begin{enumerate}
  \item The CPs and the ISP set their respective access prices $p^c_i$ and $p^s$ simultaneously.
  \item The internauts react to the prices and set their demands.
  \item The regulator sets the payment $p^d_i$ from the CP $i$ to the ISP. This can be positive or negative. Yet again, the regulator shall be able to bring only that revenue into consideration which is generated by internauts connected to CP $i$.
\end{enumerate}

The case when $\beta = 0$ is easily handled in either scenario. The actions of the various CPs (prices) do not influence each other. Though the ISP's utility is the sum over all revenues accrued from access to each CP, in setting the $p^d_i$ the regulator takes into account only the revenue generated by accesses to content of CP $i$. The ISP's utility is thus separable, and the problem separates into $n$ single-CP single-ISP problems. The results of Theorems \ref{thm:1cp-bargainbefore} and \ref{thm:1cp-bargainafter} then apply. We shall henceforth assume that $\beta > 0$. 

\subsection{Ex ante regulation}
\label{subsec:bargain-before-multicp}

Here, we characterize only those equilibrium prices that result in strictly positive demands for content from all the CPs. It is however possible that there are equilibria with some demands being zero. Using the general demand function (\ref{eqn:GeneralDemand}), we characterize all such equilibria in Sections \ref{app:MultiCPAllZeroDemand} and \ref{app:MultiCPMixedDemand} of Appendix \ref{app:GeneralDemandFunction}.

Recall the definition of the matrix $A_n$ and the vector $E_n$, given after (\ref{eqn:price-constraint-ub}). The matrix $A_n$ has diagonal entries $\alpha$ and off-diagonal entries $-\beta$. $E_n$ is the $n \times 1$ vector of all 1s.

\begin{thm}
\label{thm:multicp-bargainbefore}
  Assume $\alpha > (n-1) \beta > 0$ and consider the \exante regulation game. Among profiles with strictly positive demand, a strictly positive pure strategy Nash equilibrium exists if and only if the matrix $(A_n + 2 \alpha I_n)^{-1} [D_0 E_n + A_n p^a]$ is made of strictly positive entries. When this condition holds, the pure strategy Nash equilibria have the following properties.
  \begin{itemize}
    \item The uniqueness is up to a free choice of the vector $p^d$.
    \item At equilibrium, for each $i$, there exist constants $g_i$ and $h_i$ that depend only on $p^a,~D_0,~\alpha,~\beta$ such that
    \begin{eqnarray*}
      p^s_i & = & g_i - p^d_i \\
      p^c_i & = & h_i + p^d_i.
    \end{eqnarray*}
    \item For each CP $i$, the net internaut payment per unit demand is unique and is given by $p^s_i + p^c_i = g_i + h_i$. Any payment $p^d_i$ paid by CP $i$ is collected from the internauts, and this in turn is returned to the internauts by the ISP.
    \item The demand vector is unique and does not depend on $p^d$.
    \item The revenues per unit demand, and therefore the total revenues collected by the CPs and the ISP, does not depend on $p^d$.
  \end{itemize}
\end{thm}

The recipe for the proof is identical to that of Theorem \ref{thm:1cp-bargainafter}, only with some matrix algebra. See Appendix \ref{app:proof-bargain-before}.

{\em Remarks}: 1) Yet again we notice that the actual choice of $p^d$ does not affect the net cost per unit demand to the internauts; neither does it affect the equilibrium demand. It merely affects the way in which the payment by the internauts is split between CP $i$ and the ISP. The mere fact that the regulator fixed some $p^d$ (under \exante regulation) suffices to get an equilibrium more favorable than the case when $p^d$ is under the control of one of the players and is jointly set with that player's access price \cite{nn}.

2) For concreteness, we give specific results for the case when $n = 2$; see (\ref{eqn:solution-multicp}) in Appendix \ref{app:proof-bargain-before}. Let $\tau = \beta/\alpha$. The negative definiteness condition is then $\tau < 1$, and thus $\tau \in (0,1)$. The equilibrium prices turn out to be
      \begin{eqnarray}
        \label{eqn:multi-cp-ps}
        p^s & = & - p^d + \frac{1}{3(1 - \tau^2/9)} \left[ \begin{array}{ccc} 1 & \tau/3 \\ \tau/3 & 1 \end{array} \right] p^a + \frac{D_0}{3 \alpha (1-\tau) (1 - \tau/3)} E_2, \\
        \label{eqn:multi-cp-pc}
        p^c & = & p^d - \frac{2}{3(1 - \tau^2/9)} \left[ \begin{array}{ccc} 1 & \tau/3 \\ \tau/3 & 1 \end{array} \right] p^a + \frac{D_0}{3 \alpha (1 - \tau/3)} E_2.
      \end{eqnarray}
      An interesting observation from (\ref{eqn:multi-cp-ps}) is that when $\tau \lesssim 1$, any increase in CP 2 price causes a reduction in demand for that content, but results in nearly similar in magnitude increase in demand for content 1, and vice-versa. The ISP resources thus remain nearly fully utilized which encourages the ISP to charge a high price, as evidenced by the appearance of $1 - \tau$ in one of the denominators for $p^s$. The price charged by the CP in (\ref{eqn:multi-cp-pc}) remains bounded.

\subsection{Ex post regulation}
\label{subsec:bargain-after-multicip}

As done previously, the ISP and the CPs will choose their respective prices knowing that the revenue they earn will depend on the side payment set by the regulator. We shall present our results for $n=2$, due to combinatorial complexity reasons.

As in the $n=1$ case, the ISP and CP $i$ will share $p^s_i + p^c_i + p^a_i$, the revenue coming from internauts accessing content from CP $i$, in the proportion $\gamma_i$ and $1-\gamma_i$. One immediate observation is that at equilibrium, this revenue should be nonnegative if demand is strictly positive because otherwise CP $i$ can raise price and force demand to be zero, change his loss to zero, and strictly improve. Another observation is that all utilities and the constraints depend on $p^s_i$ and $p^c_i$ only through the sum $p^s_i + p^c_i$. While this sum is bounded if the demand vector is to be strictly positive, neither $p^s_i$ nor $p^c_i$ need be bounded, and so the action sets for each of the agents is unbounded. We shall present our main result for the \expost regulation game for $n=2$ and under a condition on the relative weights, namely, the matrix $H$ with entries
\begin{equation}
  \label{eqn:condition-on-bargains}
  H_{ij} = \left\{
    \begin{array}{ll}
      \gamma_i & i = j, \\
      -(\frac{\beta}{\alpha}) \left( \frac{\gamma_i + \gamma_j}{2} \right) & i \neq j,
    \end{array}
  \right.
\end{equation}
is positive definite. This condition arises to keep the utility of the ISP a concave function of $p^s$ in Region 1 of Figure \ref{fig:Demand Region}.

\begin{proposition}
  Let $\tau = \beta/\alpha$ and $n=2$. Then $H$ is positive definite if and only if
  \begin{equation}
    \label{eqn:condition-H-pd}
    \sqrt{\max \left\{ \frac{\gamma_1}{\gamma_2}, \frac{\gamma_2}{\gamma_1} \right\}} \leq \frac{1+\sqrt{1 - \tau^2}}{\tau}.
  \end{equation}
  Under this condition, the Hessian of $U_{\isp}$ in Region 1, given by $-2\alpha H$, is negative definite, and so $U_{\isp}$ is a concave function of $(p^s_1, p^s_2)$ in Region 1.
\end{proposition}

\begin{proof}
  $H$ is a $2 \times 2$ matrix and the statement is straightforward to verify by direct evaluation of eigenvalues and requiring that they be positive. The expression for $U_{\isp}$ immediately yields that the Hessian is $-2\alpha H$. We omit the details.
\end{proof}

This condition (\ref{eqn:condition-H-pd}) holds, for example, when the $\gamma_i$'s are equal and $\alpha > \beta$.

Our main result of this section is the following mixed bag. Recall that the case $\beta = 0$ was already considered and disposed; so we shall consider only $\beta > 0$.

\begin{thm}
\label{thm:multicp-bargainafter}
Consider $n=2$. Let the matrix $H$ given by (\ref{eqn:condition-on-bargains}) be positive definite. Also let $\alpha > \beta > 0$. Without loss of generality, assume $p^a_1 \geq p^a_2$. For the \expost regulation game, the following hold.
\begin{itemize}
  \item If $p^a_1$ is large enough so that
  \begin{equation}
    \label{eqn:too-high-pa}
    p^a_1 \geq (2 \alpha / \beta) p^a_2 + (2 \alpha / \beta - 1) D_0,
  \end{equation}
  then there exists a pure strategy Nash equilibrium with $d_1 > 0$ and $d_2 = 0$. Such an equilibrium satisfies all the properties of a single-CP and single-ISP equilibrium given in Theorem \ref{thm:1cp-bargainafter} with $D_0$ and $\alpha$ replaced by $D'_0 = D_0 (\alpha + \beta) / \alpha$ and $\alpha' = (\alpha^2 - \beta^2)/\alpha$. There is no other pure strategy Nash equilibrium.
  \item If (\ref{eqn:too-high-pa}) does not hold, there exists no pure strategy Nash equilibrium.
\end{itemize}
\end{thm}

See Appendix \ref{app:proof-bargain-after} for a proof.

Thus even though \expost regulation in the single-CP single ISP case always gave a unique Nash equilibrium with the desirable strictly positive demand, the desirable feature disappears when there are multiple CPs,  $\alpha > \beta > 0$, and $p^a_1$ is not high enough to satisfy (\ref{eqn:too-high-pa}). In particular, when $p^a_i$ are equal, there is no equilibrium in the \expost regulation game. {\em Ex ante} regulation continues to yield a unique Nash equilibrium among those profiles with strictly positive demand vectors. In the case when there is indeed an \expost regulation equilibrium, under (\ref{eqn:too-high-pa}), CP 2 is shut out by CP 1. The above result corrects an error in \cite[Th. 4]{NETCOOP10_NonneutralNetwork_AltmanHanawal} where the equilibrium under (\ref{eqn:too-high-pa}) was missed.

\section{Exclusive Contract}
\label{sec:ExclusiveContract}
In this section we study the \exante regulation game in a setting where one of the CPs (say CP 1) enters into an exclusive contract with the ISP\footnote{The same situation also arises when the ISP itself is also a provider of content.}. We refer to the pair of ISP and the colluding CP 1 as the {\em colluding pair} and denote it $\ispbar$. They make a joint decision on the prices $p_1=p_1^s+p_1^c$ charged to the internauts of class 1. We denote the sum of their utilities as $U_{\ispbar}$. The objective of the colluding pair is to maximize their joint revenue. How they share it between themselves shall not concern us here.

The utilities obtained by the  $n$ players in the resulting game are as follows:
\begin{eqnarray*}
  \label{eqn:ISPUtilityCollusion}
  U_{\ispbar}(p_1^s+p_1^c,p_2^s,\cdots, p_n^s,p_2^c,\cdots,p_n^c) & = & \Big[ D_0-\alpha (p_1^s+p_1^c) +\beta \sum_{j\neq 1}(p_j^s+p_j^c) \Big] (p_1^s+p_1^c + p_1^a) \\
  & & + \sum_{i\neq 1} \Big[ D_0 -\alpha (p_i^s +p_i^c)+  \beta \sum_{j\neq i}(p_j^s + p_j^c) \Big] (p_i^s + p_i^d),
\end{eqnarray*}
and for $i=2,3, \cdots, n$,
\begin{eqnarray*}
  \label{eqn:CPUtilityCollusion}
  U_{CP,i}(p_1^s+p_1^c,p_2^s,\cdots, p_n^s,p_2^c,\cdots,p_n^c) = \Big[ D_0- \alpha (p_i^s+p_i^c) + \beta \sum_{j\neq i}(p_j^s+p_j^c) \Big] (p_i^c + p_i^a-p_i^d).
\end{eqnarray*}
It is easy to verify that $U_{\ispbar}$ is a concave function of $\overline{p}^s:=(p_1^s+p_1^c,p_2^s,\ldots, p_n^s)$ for a given $p^a:=(p_1^a,p_2^a,\ldots, p_n^a)$ and $\overline{p}^d:=(p_2^d,p_3^d, \ldots, p_n^d)$, and for each for $i=2,3,\cdots,n$, $U_{\cp,i}$ is a concave function of $p_i^c$. Indeed, the Hessian of $U_{\ispbar}$ is $-2A_n$ which is negative semidefinite (negative definite when $\alpha > (n-1) \beta$).

The following theorem establishes the existence of equilibrium prices and studies some of its properties.

\begin{thm}
\label{thm:NEExclusiveContract}
Let $\alpha > (n-1) \beta$. In the \exante regulation game with collusion, there is a unique Nash equilibrium with the following properties.
\begin{itemize}
\item The equilibrium is unique up to a free choice of $\overline{p}^d$,
\item The equilibrium price set by the colluding pair is
\begin{equation}
\label{eqn:CollusionPriceColludingPair}
p_1^s + p_1^c= \frac{-p_1^a}{2}+ \frac{D_0}{2(\alpha-(n-1)\beta)},
\end{equation}
and for $i=2,3,\dots, n,$ the equilibrium prices are
\begin{equation}
\label{eqn:CollusionPriceNonColludingPlayers}
p_i^s=g_i-p_i^d \text{\;\;\;and\;\;\;} p_i^c=h_i+p_i^d,
\end{equation}
where the constants $g_i$ and $h_i$ depend only on $p^a, D_0, \alpha$, and $\beta$.
\item The demand vector, the revenue per unit demand, and therefore the total revenues collected by $\ispbar$ and the other $\cp$s do not depend on $\overline{p}^d$.
\end{itemize}
\end{thm}

See Appendix \ref{app:proof-nonneutrality-collusion} for a proof.

{\em Remarks}: 1) From (\ref{eqn:CollusionPriceColludingPair}), the equilibrium price charged by the colluding pair depends on only its advertisement revenues and is independent of other price quantities. The more the number of CPs, the higher the price charged by the colluding pair on internauts of class 1.

2) As before, $\overline{p}^d$ has no influence on the prices per unit demand seen by the internauts.

\subsection{Price of Collusion}
\label{sec:PriceOfCollusion}
Is collusion beneficial to the colluding pair? How does it affect the noncolluding CPs? We specialize Theorem \ref{thm:NEExclusiveContract} to case of two CPs, exploit the tractability of our model, and provide some explicit answers.

It can be straightforwardly checked that the equilibrium prices of Theorem \ref{thm:NEExclusiveContract} are
\begin{eqnarray}
  \left[ \begin{array}{c} p^s_1 + p^c_1 \\ p^s_2 \\ p^c_2 \end{array} \right]
  = \left[ \begin{array}{ccc} 0 & -1/2 & 0 \\ -1 & \tau/6 & 1/3 \\ 1 & -\tau/3 & -2/3 \end{array} \right]
    \cdot \left[ \begin{array}{c} p^d_2 \\ p^a_1 \\ p^a_2 \end{array} \right]
    + \frac{D_0}{6 \alpha} \left[ \begin{array}{c} 3/(1 - \tau) \\ (2+\tau)/(1 - \tau) \\ 2  \end{array} \right],
  \label{eqn:CollusionEquilibriumPrices}
\end{eqnarray}
where $\tau = \beta/\alpha \in [0,1)$. The net price per unit demand is then
\begin{eqnarray}
  \left[ \begin{array}{c} p^s_1 + p^c_1 \\ p^s_2 + p^c_2 \end{array} \right]
   = \left[ \begin{array}{ccc} -1/2 & 0 \\ -\tau/6 & -1/3 \end{array} \right]
    \cdot \left[ \begin{array}{c} p^a_1 \\ p^a_2 \end{array} \right]
    + \frac{D_0}{6 \alpha (1 - \tau)} \left[ \begin{array}{c} 3 \\ 4-\tau  \end{array} \right],
  \label{eqn:CollusionEquilibriumSumPrices}
\end{eqnarray}
independent of $p^d_2$ as we had observed earlier.

When there is no collusion between the ISP and CP 1, the equilibrium prices in (\ref{eqn:multi-cp-ps}) and (\ref{eqn:multi-cp-pc}) yield
\begin{eqnarray}
  \left[ \begin{array}{c} p^s_1 + p^c_1 \\ p^s_2 + p^c_2 \end{array} \right]
  = -\frac{1}{3(1 - \tau^2/9)} \left[ \begin{array}{ccc} 1 & \tau/3 \\ \tau/3 & 1 \end{array} \right]
    \cdot \left[ \begin{array}{c} p^a_1 \\ p^a_2 \end{array} \right]
    + \frac{2D_0 (1-\tau /2)}{3 \alpha (1-\tau) (1 - \tau/3)} E_2.
  \label{eqn:noCollusionEquilibriumSumPrices}
\end{eqnarray}

\cite{NETGCOOP11_RevisitingCollusion_AltmanKamedaHayel} proposed two relevant performance measures: the individual single collusion price (ISCP) and single collusion externality price (SCEP). When there is only one coalition, the ISCP is defined as the ratio of the total utilities of the colluding players before and after collusion. The SCEP is defined as the ratio of the total utilities of the noncolluding players before and after collusion. Let $(p^s,p^c)$ and $(\overline{p}^s, \overline{p}^c)$ denote the equilibrium prices under no collusion and under collusion, respectively. Then
\begin{eqnarray}
\label{eqn:ISCP_ISP}
\text{ISCP} & = & \frac{U_{\isp}(p^s,p^c)+U_{\cp,1}(p^s,p^c)}{U_{\overline{\isp}}(\overline{p}^s,\overline{p}^c)}, \\
\label{eqn:ISCP_CP2}
\text{SCEP} & = & \frac{U_{\cp,2}(p^s,p^c)}{U_{\cp,2}(\overline{p}^s,\overline{p}^c)}.
\end{eqnarray}
Substitution of (\ref{eqn:CollusionEquilibriumSumPrices}) and (\ref{eqn:noCollusionEquilibriumSumPrices}) in (\ref{eqn:ISCP_ISP}) and (\ref{eqn:ISCP_CP2}) provides a wealth of information:

\begin{enumerate}
  \item $\text{SCEP} \geq 1$ if and only if
  \begin{equation}
    \label{eqn:CollusionBenefitThresold}
    p_2^a \leq \frac{3-\tau^2}{2\tau}p_1^a+\frac{D_0(3+\tau)}{\alpha(2\tau)}.
  \end{equation}
  When (\ref{eqn:CollusionBenefitThresold}) holds, collusion between the ISP and CP 1 results in a loss in revenue for CP 2.
  \item Under (\ref{eqn:CollusionBenefitThresold}), both classes of internauts pay lower prices. But the demand for CP 2 content is lower.
  \item For a specific choice of parameters, see Figure \ref{fig:ISCP_ISP} for a plot of ISCP and SCEP versus the advertisement revenue parameter for CP 2. The colluding pair benefits exactly when the noncolluding CP has lower revenue.
   \item The colluding pair does not always benefit, for e.g., when (\ref{eqn:CollusionBenefitThresold}) is violated. See also \cite{NETGCOOP11_RevisitingCollusion_AltmanKamedaHayel}.
\end{enumerate}
We refer the reader to \cite[Sec.~4]{WPIN12_GameTheoreticAnalysis_HanawalAltmanSundaresan} for more details.

\begin{figure}[tb]
\begin{centering}
  \includegraphics[scale=.4]{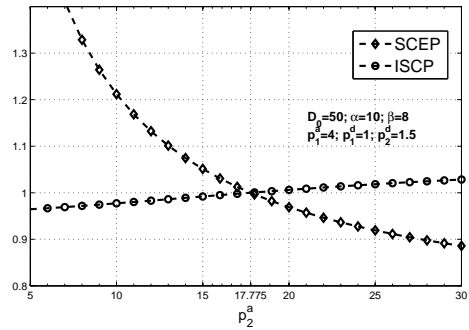}
  \caption{Individual Single collusion price for ISP}\label{fig:ISCP_ISP}
\end{centering}
\end{figure}

\section{Discussion}
\label{sec:discussion}
In this paper, we studied a model of a nonneutral network with off-network pricing and investigated mechanisms for regulating the payments between the ISP and the CPs. Our main observation is that a mild form of regulation can bring benefits to all players. We proposed two mechanisms based on a weighted proportional sharing criterion to regulate the side payments between the ISP and the CPs. In the \exante regulation game, where the regulator decides the side payment before the access prices are set competitively, the internauts do not get affected by the actual regulated prices between the ISP and the CPs. In particular, when the access prices result in positive demand, the equilibrium demand and the access prices are the same as in the case of the no off-network pricing (when $p^d = 0$). From the internauts' perspective, the mere presence of a regulator who regulates side payments in the \exante regulation game makes the internauts indifferent to the neutral and the nonneutral regimes. In the \expost regulation game, where the regulator sets the side payments after the prices are set competitively, price competition can result in zero demand equilibria when there are multiple CPs. All these observations appear to tilt the balance in favor of \exante regulation, though in the single-CP single-ISP setting it leads to higher prices.

In the nonneutral regime vertical monopolies can be formed. We considered a simple case of vertical monopoly where CP 1 colludes with the ISP. Such a collusion is beneficial to the CP only if its advertising revenues are higher than a certain threshold.

To keep our analysis tractable, we have used linear demand functions that are popular in the inventory management literature. The biggest benefit of using linear demand function is that it is tractable, as evidenced by the obtained expressions in this paper. It is naturally of interest to see the extent to which more nuanced demand functions may change the qualitative conclusions obtained in our paper.



\elecappendix

\medskip

\section{Derivation of General Demand Function}
\label{app:GeneralDemandFunction}
\setcounter{section}{1}
With a suitable reindexing, we may assume that the vector $p = p^s + p^c$ has components in the increasing order, i.e., $p_1 \leq p_2 \leq \cdots \leq p_n$, where $p_i = p_i^s + p_i^c$. For brevity, we shall abuse notation and refer to $d_i(p^s, p^c)$ as $d_i(p)$. Common sense suggests that if demand for CP $i$ content is zero, then demand for CP $j$ content for a $j \geq i$ must also be zero since its price is higher. It will be illuminating to study the evolution of the demand function as the price vector increases from the all-zero vector to $p$ via $\min\{ xE_n, p\}$, where $x$ is a scalar parameter that increases from 0 to $+\infty$ and the $\min$ operation is taken component-wise.

For $x \in [0,p_1]$, we have $\min\{ xE_n, p\} = xE_n$; all internauts are charged the same (net) price of $x$ per unit demand. It is then immediate that all demands are equal, and from (\ref{eqn:demand-multi-CP}), this value is strictly positive if and only if
\[
  x < \frac{D_0}{\alpha - (n-1)\beta}.
\]
In particular, demand for CP 1 is strictly positive at $x = p_1$ if and only if
\begin{equation}
  \label{eqn:PositveDemandCondnCP1}
  p_1 < \frac{D_0}{\alpha - (n-1)\beta} =: T(1).
\end{equation}
If (\ref{eqn:PositveDemandCondnCP1}) does not hold, the demand for the cheapest content is zero, and our common sense conclusion suggests that all other demands are also zero. If (\ref{eqn:PositveDemandCondnCP1}) holds, then at $x = p_1$, demand for CP 1 is strictly positive. For $x \in [0, p_1)$, the demand $d_1$ for content from CP 1 decreased with $x$. But further increase in $x$ leaves the price for CP 1 content unchanged at $p_1$, and our observations about positive correlation with respect to others' prices indicates that $d_1$ must now begin to linearly increase with $x$ for $x > p_1$. This is illustrated in Figure \ref{fig:Demand}. Thus for $x \in [p_1, p_2]$, we see
\begin{eqnarray}
  d_1 & = & D_0 - \alpha p_1 + (n-1) \beta x, \text{ for CP } 1 \nonumber \\
  \label{eqn:CP2-Demand}
  d_i & = & D_0 - (\alpha -(n-2) \beta) x + \beta p_1, \text{ for CP } i \geq 2. \quad \quad
\end{eqnarray}
At $x = p_2$, the demand from CP 2, given by (\ref{eqn:CP2-Demand}) for $i=2$, is positive if and only if
\begin{equation}
  \label{eqn:CP2-PositiveDemand-Condn}
  p_2 < \frac{D_0 + \beta p_1}{\alpha-(n-2)\beta} =: T(2).
\end{equation}
When (\ref{eqn:CP2-PositiveDemand-Condn}) holds, $d_1$ is linear in $x$ with positive slope $(n-1) \beta$ for $x$ up to $p_2$, and all other $d_i$ are linear and decreasing in $x$ with negative slope $-(\alpha - (n-2)\beta)$. Again see Figure \ref{fig:Demand}. If (\ref{eqn:CP2-PositiveDemand-Condn}) does not hold, $d_i = 0$ for $i \geq 2$, but $d_1$ is set up to the value $D_0 - \alpha p_1 + (n-1) \beta x^*$ where $x^* = T(2)$. All demands are thus set in this latter case. If (\ref{eqn:CP2-PositiveDemand-Condn}) holds, the former case, then one proceeds further in a similar fashion until $x^* = p_n$ and all demands are set, or until $x^* \in (p_{k^*}, p_{k^* + 1}]$ for some $k^*$, when demands $d_j = 0$ for all $j \geq k^*$, and demands $d_i$ are set with prices $\min \{ x^* E_n, p \}$. To get an explicit expression for the demands, let us define
\begin{equation}
  \label{eqn:DemandThresholdDefinition}
  T(k) := \frac{D_0 + \beta \sum_{j: j < k} p_j}{\alpha - (n-k) \beta}, \quad k = 1, 2, \ldots, n.
\end{equation}
Let $k^*$ be the smallest index among $k = 0, 1, \ldots, n$ for which
\begin{eqnarray}
  p_i & < & T(i), \quad i = 1, \ldots, k \nonumber \\
  \label{eqn:ZeroDemandCondn}
  p_{k+1} & \geq & T(k+1).
\end{eqnarray}
To further clarify (\ref{eqn:ZeroDemandCondn}), if $p_1 \geq T(1)$ then $k^* = 0$; if $p_i < T(i)$ for $i = 1, \ldots, n$, then $k^* = n$. In all other cases, the definition in (\ref{eqn:ZeroDemandCondn}) is unambiguous. Straightforward manipulations show that
\[
  T(k) > T(k+1) \mbox{ if and only if } p_k < T(k), ~ k = 1, \ldots, n-1.
\]
It follows from the definition of $k^*$ that
\begin{equation}
  \label{eqn:T-ordering}
  T(1)>T(2)>\cdots>T(k^*)>T(k^*+1)\leq T(k^*+2),
\end{equation}
where the last two inequalities hold if the corresponding indices are between 1 and $n$. Let us now get back to identifying the demands. Given $k^*$, we set $x^*$ such that
\[
  D_0 - \alpha x^* + \beta \sum_{j < k^* + 1} p_j + \beta (n - k^* - 1) x^* = 0;
\]
the solution is $x^* = T(k^* + 1)$. The demands are now specified by
\begin{equation}
  d_i(p) = \left\{
    \begin{array}{ll}
      D_0 - \alpha p_i + \beta \sum_{j < k^*+1, j \neq i} p_j + (n-k^*) \beta T(k^* + 1), & i = 1, \ldots, k^* \\
      0, & i > k^*.
    \end{array}
    \right.
\end{equation}

\begin{figure}[tb]
\begin{centering}
\includegraphics[width=3.45in, height=2in]{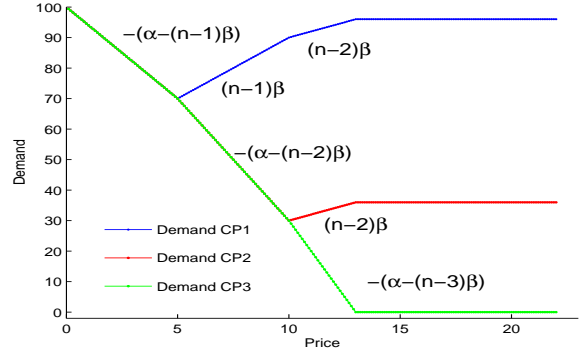}
\caption{Demand functions. Abscissa is the parameter $x$.}
\label{fig:Demand}
\end{centering}
\end{figure}

This describes the behavior of the internauts for any given price vector $p = p^s + p^c$ and models the positive correlation of demand with other internauts' prices. Figure \ref{fig:Demand} depicts the procedure outlined above to evaluate the demands when there are $n=3$ CPs. The other parameters are $D_0=100$, $\alpha=10$, $\beta=2$, and the price vector $p=(5,10,20)$. The slope of demand functions in different intervals are also marked. Here $k^* = 2$. The demand of each CP is obtained by noting the respective value of their demand curve at $x^*=T(3)$.

\subsection{Equilibria with all demands being zero}
\label{app:MultiCPAllZeroDemand}
We now study the case of equilibria with all demands being zero. Obviously (\ref{eqn:PositveDemandCondnCP1}) must not hold; additional conditions are also needed.

\begin{thm}
  \label{thm:multi-cp-bb-zerodemand}
  A price vector $(p^s, p^c)$ is an equilibrium with all demands being zero if and only if the following conditions hold for all $i = 1, 2, \ldots, n$:
  \begin{eqnarray}
    \label{eqn:all-zero-high-ps}
    p^s_i & \geq & \frac{D_0}{\alpha - (n-1) \beta} + p^a_i - p^d_i \\
    \label{eqn:all-zero-high-pc}
    p^c_i & \geq & \frac{D_0}{\alpha - (n-1) \beta} + p^d_i.
  \end{eqnarray}
\end{thm}

\begin{proof}
  See Appendix \ref{app:proof-thm:multi-cp-bb-zerodemand}.
\end{proof}

{\em Remarks}: 1) Equations (\ref{eqn:all-zero-high-pc}) and (\ref{eqn:all-zero-high-ps}) are the same as saying that revenues per unit demand due to CP $i$ content, to CP $i$ and to the ISP, are at least
\[
  \frac{D_0}{\alpha - (n-1) \beta} + p^a_i = T(1) + p^a_i,
\]
and this holds for each $i$. In such a case, all the CPs and the ISP are charging too high a price resulting in a deadlock equilibrium with all demands zero.

2) When $n=1$, (\ref{eqn:all-zero-high-ps}) and (\ref{eqn:all-zero-high-pc}) reduce to (\ref{eqn:zero-demand-eq1}) and (\ref{eqn:zero-demand-eq2}), as they should.

\subsection{Equilibria with mixed demands with $n=2$}
\label{app:MultiCPMixedDemand}
In order to avoid combinatorial complexities, and for ease of exposition, we focus on the case when $n=2$ and now characterize all equilibria where demand for one content is strictly positive and demand for the other is zero.

\begin{thm}
\label{thm:multicp-bb-mixed}
(a) A price profile $((p^s_1, p^s_2), p^c_1, p^c_2)$ is an equilibrium with $d_1 > 0$ and $d_2 = 0$ if and only if
\begin{eqnarray}
  p^s_1 & = & \frac{D'_0 + \alpha' p^a_1}{3 \alpha'} - p^d_1 \label{eqn:ps1-mixed}\\
  p^c_1 & = & \frac{D'_0 - 2 \alpha' p^a_1}{3 \alpha'} + p^d_1 \label{eqn:pc1-mixed}\\
  p^s_2 & \geq & \frac{D_0 + \beta (p^s_1 + p^c_1)}{\alpha} - p^d_2 + p^a_2 \label{eqn:ps2-mixed} \\
  p^c_2 & \geq & \frac{D_0 + \beta (p^s_1 + p^c_1)}{\alpha} - p^s_2 \label{eqn:pc2-mixed},
\end{eqnarray}
where $D'_0 = D_0(\alpha+\beta)/\alpha$ and $\alpha' = (\alpha^2 - \beta^2)/\alpha$.

(b) A price profile $((p^s_1, p^s_2), p^c_1, p^c_2)$ is an equilibrium with $d_2 > 0$ and $d_1 = 0$ if and only if the same conditions as above hold with indices 1 and 2 interchanged.
\end{thm}

\begin{proof}
  See Appendix \ref{app:proof-thm:multicp-bb-mixed}.
\end{proof}

{\em Remarks}: 1) Region 1 equilibria are characterized in Theorem \ref{thm:multicp-bargainbefore}. Region 2 equilibria are characterized in Theorem \ref{thm:multi-cp-bb-zerodemand}. Equilibria in Regions 3 and 4 are characterized in Theorem \ref{thm:multicp-bb-mixed}. We have therefore characterized all equilibria in the $n=2$ case.

2) Conditions (\ref{eqn:ps1-mixed}) and (\ref{eqn:pc1-mixed}) together constitute an equilibrium in case of a single CP with $D_0$ and $\alpha$ replaced by $D'_0$ and $\alpha'$, respectively.

3) Conditions (\ref{eqn:ps2-mixed}) and (\ref{eqn:pc2-mixed}) may be interpreted as
\begin{eqnarray*}
  r_{\isp,2} \geq T(2) + p^a_2 \mbox{ and } p_2 \geq T(2)
\end{eqnarray*}
where $r_{\isp,2} = p^s_2 + p^d_2$ is the revenue to the ISP from CP 2 content.

\section{Proof of Theorem \ref{thm:multi-cp-bb-zerodemand}}
\label{app:proof-thm:multi-cp-bb-zerodemand}
We first prove the necessity of these conditions. Let $(p^s, p^c)$ be an equilibrium with all demands being zero; it must be the case that (\ref{eqn:PositveDemandCondnCP1}) is violated, and so
\[
  p^s + p^c \geq \left( \frac{D_0}{\alpha - (n-1) \beta} \right) E_n.
\]
CP $i$ should not be able to reduce his price, increase demand $d_i$ to a strictly positive value, and derive a strictly positive utility. It must therefore be the case that even the least reduced price $q^c_i$ that keeps the demand $d_i$ on the threshold of positivity is too low a price bringing him a negative revenue. More precisely, a price
\begin{equation}
  \label{eqn:cp-least-reduced-price}
  q^c_i + p^s_i = \frac{D_0}{\alpha - (n-1) \beta}
\end{equation}
when demand for CP $i$ content is on the threshold of positivity should imply
\begin{equation}
  \label{eqn:cp-neg-revenue}
  q^c_i - p^d_i + p^a_i \leq 0,
\end{equation}
a negative revenue per unit demand for CP $i$. Substitution of (\ref{eqn:cp-least-reduced-price}) in (\ref{eqn:cp-neg-revenue}) yields necessity of (\ref{eqn:all-zero-high-ps}).

We next verify necessity of (\ref{eqn:all-zero-high-pc}) by contraposition. Let $i$ be a content index for which (\ref{eqn:all-zero-high-pc}) does not hold, and so
\begin{equation}
  \label{eqn:pci-pdi}
  p^c_i - p^d_i < D_0 / (\alpha - (n-1) \beta).
\end{equation}
Take
\begin{equation}
  \label{eqn:epsilon}
  \varepsilon = \frac{1}{2} \left( \frac{D_0}{\alpha - (n-1) \beta} - (p^c_i - p^d_i) \right) > 0,
\end{equation}
and set $q^s_i$ so that $q^s_i + p^d_i = \varepsilon > 0$, i.e., the ISP revenue from CP $i$ content is $\varepsilon > 0$. Also set all other $p^s_j$ for $j \neq i$ to high values so that demand for these other contents is zero. Demand $d_i$ for CP $i$ content is however strictly positive because, by using $\varepsilon = q^s_i + p^d_i$, (\ref{eqn:pci-pdi}) and (\ref{eqn:epsilon}), we get
\begin{eqnarray*}
  q^s_i + p^c_i & = & (q^s_i + p^d_i) - p^d_i \\
  & = & \varepsilon + p^c_i - p^d_i \\
  & = & \frac{1}{2} \left( \frac{D_0}{\alpha - (n-1) \beta} - (p^c_i - p^d_i) \right) + (p^c_i - p^d_i) \\
  & < & \frac{D_0}{\alpha - (n-1) \beta}.
\end{eqnarray*}
Thus (\ref{eqn:PositveDemandCondnCP1}) holds, and so $d_i > 0$. (All other demands are zero). The ISP now has a strictly positive utility, and the profile cannot be a Nash equilibrium. By contraposition, necessity of (\ref{eqn:all-zero-high-pc}) is established.

We now argue sufficiency of (\ref{eqn:all-zero-high-ps}) and (\ref{eqn:all-zero-high-pc}). Take a profile that satisfies these conditions. A glance at the proof of necessity of (\ref{eqn:all-zero-high-ps}) indicates that there is no deviation for CP $i$ to derive a positive utility. To see that there is no deviation for the ISP that will yield a positive revenue, let $q^s$ be any vector of ISP prices. Without loss of generality, reorder the prices $q_i = q^s_i + p^c_i$, $i=1,2,\ldots,n$, so that $q_1 \leq q_2 \leq \cdots \leq q_n$. If ISP revenue for content $j$ were strictly positive, then
\[
  0 < q^s_j + p^d_j \leq q^s_j + p^c_j - \frac{D_0}{\alpha - (n-1) \beta}
\]
where the second inequality follows from (\ref{eqn:all-zero-high-pc}), and so
\begin{equation}
  \label{eqn:pos-revenue}
  q_j = q^s_j + p^c_j > \frac{D_0}{\alpha - (n-1) \beta} = T(1)
\end{equation}
for any such $j$ with strictly positive ISP revenue for content $j$. ($T(1)$ is defined in (\ref{eqn:PositveDemandCondnCP1})). However, from (\ref{eqn:T-ordering}) and (\ref{eqn:ZeroDemandCondn}), any index with strictly positive demand satisfies $q_i < T(i) \leq T(1)$. Comparing this with (\ref{eqn:pos-revenue}), we deduce that indices with strictly positive demand have nonpositive revenue per unit demand. The ISP revenue is therefore nonpositive, and there is no deviation that will yield a better revenue. This proves sufficiency of the stated conditions, and the characterization of all Nash equilibria with all demands being zero is complete.

\section{Proof of Theorem \ref{thm:multicp-bb-mixed}}
\label{app:proof-thm:multicp-bb-mixed}

We shall prove only (a). Proof of (b) is similar and is omitted.

We first prove necessity of the stated conditions. Let $((p^s_1, p^s_2), p^c_1, p^c_2)$ be an equilibrium with $d_1 > 0$ and $d_2 = 0$. Then, from the discussion on demands, we must have
\begin{eqnarray}
  p_1 & = & p^s_1 + p^c_1 ~ < ~ \frac{D_0}{\alpha - \beta} \nonumber \\
  p_2 & = & p^s_2 + p^c_2 ~ \geq ~ \frac{D_0 + \beta p_1}{\alpha}. \label{eqn:p2-region4}
\end{eqnarray}

Necessity of (\ref{eqn:pc2-mixed}) is immediate from (\ref{eqn:p2-region4}).

We next prove the necessity of (\ref{eqn:ps2-mixed}). Since $d_2 = 0$, the current utility for CP 2 is zero. No unilateral deviation of CP 2 should yield him a strictly positive utility. For a strictly positive utility, he must reduce his price to make the demand for his content strictly positive. But even the least reduction in his price that puts the demand for his content on the threshold of positivity, a price $q^c_2$ such that $q^c_2 + p^s_2 = (D_0 + \beta p_1) / \alpha$ should already yield him a net nonpositive revenue $q^c_2 + p^a_2 - p^d_2 \leq 0$. Substitution of the former equality in the latter inequality yields (\ref{eqn:ps2-mixed}) as a necessary condition.

We now consider deviations of the ISP. We first observe that ISP's utility must be nonnegative. Next, given that the price profile falls in Region 4 of Figure \ref{fig:Demand Region}, the ISP can reduce the price of $p^s_2$ to $q^s_2$ such that
\begin{equation}
  \label{eqn:q2}
  q_2 = q^s_2 + p^c_2 = T(2) = \frac{D_0 + \beta p_1}{\alpha}
\end{equation}
without affecting the demand $d_1$ and keeping the demand $d_2 = 0$. His revenue does not change, and the price profile $(p_1, q_2)$ is now on the line BO in Figure \ref{fig:Demand Region}. The ISP's utility is thus
\begin{eqnarray}
  U_{\isp}(p_1,p_2) & = & U_{\isp}(p_1, q_2) \nonumber \\
  & = & (D_0 - \alpha p_1 + \beta q_2) (p^s_1 + p^d_1) + (D_0 - \alpha q_2 + \beta p_1) (q^s_2 + p^d_2) \nonumber \\
  & = & (D_0 - \alpha p_1 + \beta q_2) (p_1 - p^c_1 + p^d_1) + (D_0 - \alpha q_2 + \beta p_1) (q_2 - p^c_2 + p^d_2). \quad \label{eqn:U_isp-BO}
\end{eqnarray}

Let us now consider infinitesimal deviations either into Region 1 or along the line BO, and prove necessity of (\ref{eqn:ps1-mixed}) and (\ref{eqn:pc1-mixed}). The ISP can clearly change $p^s_1$ and $q^s_2$ simultaneously to place the price vector in a neighborhood of $(p_1,q_2)$ inside Region 1 or on the line BO. Such deviations are given by increments $u = (u_1, u_2)$ that satisfy $u_2 \leq (\beta / \alpha) u_1$. Since $U_{\isp}(p_1,q_2)$ given by (\ref{eqn:U_isp-BO}) is differentiable in this region, and there must be no direction pointing into Region 1 in which $U_{\isp}$ increases, we must have the dot-product
\[
  \nabla U_{\isp}(p_1,q_2)^T ~ u \leq 0 \quad \forall u \mbox{ with } u_2 \leq (\beta / \alpha) u_1.
\]
It follows that the direction of steepest ascent for $U_{\isp}$ at $(p_1,q_2)$ which is $\nabla U_{\isp}(p_1,q_2)$ must be normal to the line defined by $u_2 = (\beta/\alpha) u_1$ and pointing away from Region 1, i.e.,
\begin{equation}
  \label{eqn:partials-condition}
  \frac{\partial U_{\isp}}{\partial q_2} = - \frac{\alpha}{\beta} \frac{\partial U_{\isp}}{\partial p_1}.
\end{equation}
From (\ref{eqn:U_isp-BO}), and after noting that $r_{\isp,1} := (p_1 - p^c_1 + p^d_1)$ and $r_{\isp,2} := (q_2 - p^c_2 + p^d_2)$ are the ISP revenues per unit demand arising from contents of CP 1 and CP 2, respectively, we get
\[
  \frac{\partial U_{\isp}}{\partial q_2} = \beta r_{\isp,1} + d_2 - \alpha r_{\isp,2} =  \beta r_{\isp,1} - \alpha r_{\isp,2}
\]
and
\[
  \frac{\partial U_{\isp}}{\partial p_1} = \beta r_{\isp,2} + d_1 - \alpha r_{\isp,1}.
\]
Substitution of these in (\ref{eqn:partials-condition}) yields that the condition
\begin{equation}
  \label{eqn:Region2-necessarycondition}
  d_1 = \alpha' r_{\isp,1}
\end{equation}
is a necessary condition for direction of increase for the ISP's utility. We now use $d_1 = D_0 - \alpha p_1 + \beta q_2$, the expression for $q_2$ given in (\ref{eqn:q2}), and the definition of $r_{\isp,1}$ and rewrite (\ref{eqn:Region2-necessarycondition}) as
\begin{equation}
  \label{eqn:alt-ps1-necessity}
  D'_0 - \alpha' (p^s_1 + p^c_1) = \alpha' (p^s_1 + p^d_1).
\end{equation}
Note that this equation fixes $p^s_1$ given a $p^c_1$:
\begin{equation}
  \label{eqn:ps1-given-pc1}
  p^s_1 = \frac{D'_0 - \alpha' (p^c_1 + p^d_1)}{2 \alpha'}.
\end{equation}
Furthermore, if we can establish the necessity of (\ref{eqn:pc1-mixed}) which fixes $p^c_1$, then (\ref{eqn:alt-ps1-necessity}) implies the necessity of (\ref{eqn:pc1-mixed}) as well, as can be verified by direct substitution.

We now establish the necessity of (\ref{eqn:pc1-mixed}). Consider first an interior point of Region 4. Small deviation by CP 1 move the point along the abscissa, and if small enough the deviation keeps the resulting point inside the interior of Region 4. Then $d_2$ continues to be 0 and $d_1 > 0$. As a consequence, it follows that $d_1 = D'_0 - \alpha' p_1$, where $p_1 = p^s_1 + p^c_1$ and the variation here is in $p^c_1$. The revenue for CP 1 is $p^c_1 + p^a_1 - p^d_1$ so that
\[
  U_{\cp,1} = (D'_0 - \alpha' (p^s_1 + p^c_1)) (p^c_1 + p^a_1 - p^d_1).
\]
It is thus necessary that the first order optimality condition hold, and so
\begin{equation}
  \label{eqn:partial-Ucp1}
  \frac{\partial U_{\cp,1}}{\partial p^c_1} = (D'_0 - \alpha' (p^s_1 + p^c_1)) - \alpha' (p^c_1 + p^a_1 - p^d_1) = 0
\end{equation}
so that
\begin{equation}
  \label{eqn:alt-pc1-necessity}
  (D'_0 - \alpha' (p^s_1 + p^c_1)) = \alpha' (p^c_1 + p^a_1 - p^d_1).
\end{equation}
Solving the simultaneous equations (\ref{eqn:alt-ps1-necessity}) and (\ref{eqn:alt-pc1-necessity}), we get the necessity of (\ref{eqn:ps1-mixed}) and (\ref{eqn:pc1-mixed}) among interior points of Region 4.

Now consider points on the line BO. Let us denote the right-hand sides of (\ref{eqn:ps1-mixed}) and (\ref{eqn:pc1-mixed}) as $p^s_{1,\opt}$ and $p^c_{1,\opt}$, respectively, i.e.,
\begin{eqnarray}
  p^s_{1, \opt} & = & \frac{D'_0 + \alpha' p^a_1}{3 \alpha'} - p^d_1 \label{eqn:ps1-mixed-star}\\
  p^c_{1, \opt} & = & \frac{D'_0 - 2 \alpha' p^a_1}{3 \alpha'} + p^d_1 \label{eqn:pc1-mixed-star}
\end{eqnarray}

If $p^c_1 \geq p^c_{1,\opt}$, consider an infinitesimal decrease in $p^c_1$ which puts the point in the interior of Region 4. The left partial derivative is
\begin{equation}
  \label{eqn:partial-Ucp1-left}
  \frac{\partial^- U_{\cp,1}}{\partial p^c_1} = (D'_0 - \alpha' (p^s_1 + p^c_1)) - \alpha' (p^c_1 + p^a_1 - p^d_1),
\end{equation}
the right-hand side of (\ref{eqn:partial-Ucp1}). We then have the following chain of equalities:
\begin{eqnarray}
  \frac{\partial^- U_{\cp,1}}{\partial p^c_1} & = & (D'_0 - \alpha' (p^s_1 + p^c_1)) - \alpha' (p^c_1 + p^a_1 - p^d_1) \nonumber \\
  & = & (D'_0 - \alpha' (p^s_{1,\opt} + p^c_{1,\opt})) - \alpha' (p^c_{1,\opt} + p^a_1 - p^d_1) \nonumber \\
  & & + ~ \alpha' (p^s_{1,\opt} - p^s_1) + 2 \alpha' (p^c_{1, \opt} - p^c_1) \label{eqn:add-substract} \\
  & = & 0 - (1/2) \alpha' (p^c_{1, \opt} - p^c_1) + 2 \alpha' (p^c_{1, \opt} - p^c_1) \label{eqn:left-partial-consequence1}\\
  & = & (3/2) \alpha' (p^c_{1, \opt} - p^c_1), \label{eqn:left-partial-consequence2}
\end{eqnarray}
where (\ref{eqn:add-substract}) follows by adding and subtracting
$$(D'_0 - \alpha' (p^s_{1,\opt} + p^c_{1,\opt})) - \alpha' (p^c_{1,\opt} + p^a_1 - p^d_1).$$
Equation (\ref{eqn:left-partial-consequence1}) follows because (\ref{eqn:alt-pc1-necessity}) and ({\ref{eqn:alt-ps1-necessity}}) hold for the pair $(p^c_{1,\opt},p^s_{1,\opt})$, and from (\ref{eqn:ps1-given-pc1}) we see that $$p^s_1 - p^s_{1,\opt} = -(1/2) (p^c_1 - p^c_{1,\opt}).$$
From (\ref{eqn:left-partial-consequence2}), $p^c_1 > p^c_{1,\opt}$ implies that an infinitesimal decrease results in a strict increase for CP 1. It must therefore be that $p^c_1 \leq p^c_{1,\opt}$ for the profile under consideration to be an equilibrium.

When $p^c_1 \leq p^c_{1,\opt}$, consider an infinitesimal increase in $p^c_1$ which puts the point in the interior of Region 1, i.e., both $d_1$ and $d_2$ become positive. As a consequence, the right-derivative is now
\begin{equation}
  \label{eqn:partial-Ucp1-rigt}
  \frac{\partial^+ U_{\cp,1}}{\partial p^c_1} = (D'_0 - \alpha' (p^s_1 + p^c_1)) - \alpha (p^c_1 + p^a_1 - p^d_1);
\end{equation}
observe that the difference with (\ref{eqn:partial-Ucp1-left}) is that the second term is multiplied only by $\alpha$ instead of $\alpha'$ as now both CPs have positive demand upon deviation in Region 1. Following the same steps leading to (\ref{eqn:left-partial-consequence2}), we now get
\begin{eqnarray*}
\frac{\partial^+ U_{\cp,1}}{\partial p^c_1} & = & 0 - (1/2) \alpha' (p^c_{1, \opt} - p^c_1) + (\alpha' + \alpha) (p^c_{1, \opt} - p^c_1) \\
& = & (\alpha + \alpha'/2) (p^c_{1, \opt} - p^c_1)
\end{eqnarray*}
and now the right-hand side has a different scale when compared with (\ref{eqn:left-partial-consequence2}). When $p^c_1 < p^c_{1,\opt}$, we have $\frac{\partial^+ U_{\cp,1}}{\partial p^c_1} > 0$ yielding a strict increase in CP 1 utility. It follows that we must have $p^c_1 = p^c_{1,\opt}$. This establishes the necessity of (\ref{eqn:pc1-mixed}), and the proof of necessity is complete.

Next, to address sufficiency of the stated conditions, consider a profile satisfying them. Our necessity argument for (\ref{eqn:ps2-mixed}) also shows that CP 2 has no deviation yielding him a strictly positive utility. For the ISP, the necessity argument considered an equivalent point on the line BO, and showed that there are no infinitesimal deviations around this point that will yield a better utility. But on account of concavity of the utility functions, no other point in Region 1 (including the boundary AO) will yield a strictly better utility. Since the boundary AO has also been considered, and the any point in Region 3 yields him the same utility as the point on the line AO with the same ordinate, no point in Region 3 will yield a better utility. Similarly, on account of concavity, CP 1 too as no deviation (infinitesimal or otherwise) that will yield him a strictly better utility. This concludes the proof of sufficiency.

\section{Proof of Theorem \ref{thm:multicp-bargainbefore}}
\label{app:proof-bargain-before}
Consider a fixed $p^d$. We shall only focus on strategies jointly constrained so that $d_i > 0$ for all $i$. The joint constraint on $p^s$ and $p^c$ is given by (\ref{eqn:price-constraint-ub}), and the demands are given by (\ref{eqn:demand-multi-CP}). Let us look at $U_{\isp}$ as a function of $p^s$ and $U_{\cp, i}$ as a function $p^c_i$. We already saw that the former is concave since $\alpha > (n-1) \beta$. Inspection of the expression for $U_{\cp, i}$ shows that it is also quadratic and strictly concave in $p^c_i$. Since we seek equilibria with strictly positive demand, such equilibria are interior points of, for example in case of $n = 2$, Region 1 in Figure \ref{fig:Demand Region}. It is therefore necessary that first order optimality conditions hold for such equilibria. So, setting the gradient of $U_{\isp}$ with respect to $p^s$ to zero, we get
\begin{eqnarray*}
  \frac{\partial U_{\isp}}{\partial p^s_k}
    & = & \sum_{j : j \neq k} \beta (p^s_j + p^d_j) - \alpha (p^s_k + p^d_k) + D_0 - \alpha (p^s_k + p^c_k) + \beta \sum_{j : j \neq k} (p^s_j + p^c_j) = 0
\end{eqnarray*}
for each $k$. Similarly, setting each $\partial U_{\cp, k} / \partial p^c_k = 0$ yields
\[
  D_0 - \alpha(p^s_k + p^c_k) + \beta \sum_{j: j \neq k} (p^s_j + p^c_j) - \alpha (p^c_k + p^a_k - p^d_k) = 0
\]
for each $k$. We next write these $2n$ equations in matrix notation. For this purpose recall that the matrix $A_n = (\alpha + \beta) I_n - \beta J_n$, where all diagonal elements are $\alpha$ and all off-diagonal elements are $-\beta$, and define $B_n = (\alpha + \beta/2) I_n - (\beta/2) J_n$, where all diagonal elements are $\alpha$ and all off-diagonal elements are $-\beta/2$. Also recall that $E_n$ is the vector of size $n \times 1$ with all-one entries. Then the above equations become:
\begin{equation}
\label{eqn:multicp-statpoint}
  \left[
    \begin{array}{cc}
      2A_n & A_n  \\
      A_n  & 2B_n \\
    \end{array}
  \right]
  \left[ \begin{array}{c}p^s \\ p^c \end{array} \right]
  =
  \left[
    \begin{array}{cc}
      -A_n & \bigcirc \\
      \alpha I_n & - \alpha I_n
    \end{array}
  \right]
  \left[ \begin{array}{c}p^d \\ p^a \end{array} \right] + D_0 E_{2n}.
\end{equation}
The matrices $A_n$ and $B_n$ commute because both are linear combinations of the commuting matrices $I_n$ and $J_n$. Moreover, the determinant of the matrix on the left side is
\begin{eqnarray*}
  \det(4A_nB_n - A_n^2) & = & \det(A_n (A_n+2\alpha I_n)) \\
  & = & \det(A_n) \det(A_n+2\alpha I_n) \\
  & = & (\alpha+\beta)^{n-1}(\alpha -(n-1)\beta)(3 \alpha + \beta)^{n-1} (3\alpha - (n-1)\beta) \\
  & > & 0.
\end{eqnarray*}
This follows because the eigenvalues of the matrix
$$M(\rho) = (1-\rho)I_n + \rho J_n$$
are $1-\rho$ repeated $n-1$ times and $1 + (n-1)\rho$ occurring once. The matrices $A_n$ and $A_n + 2 \alpha I_n$ are scaled versions of $M(\rho)$ with appropriate choices for $\rho$. Thus the matrix on the left side of (\ref{eqn:multicp-statpoint}) is invertible. From the fact that $A_n$ and $B_n$ commute, the fact that $4A_nB_n - A_n^2 = A_n(A_n+2\alpha I_n)$, and the formula for the inverse of two-by-two block matrices with commutable entries, one writes by inspection that
\[
  \left[
    \begin{array}{cc}
      2A_n & A_n  \\
      A_n  & 2B_n \\
    \end{array}
  \right]^{-1}
  = (A_n(A_n+2\alpha I_n))^{-1} \circ
  \left[
    \begin{array}{cc}
      2B_n & -A_n  \\
      -A_n & 2A_n \\
    \end{array}
  \right],
\]
where the symbol ``$\circ$'' implies that the matrix before it left-multiplies all the elements of the bigger matrix following it. Multiplying (\ref{eqn:multicp-statpoint}) by the above inverse, and observing that $2B_n + \alpha I_n = A_n + 2 \alpha I_n$, we get
\begin{eqnarray}
  \left[ \begin{array}{c}p^s \\ p^c \end{array} \right]
  =
  \left[
    \begin{array}{cc}
      -I_n & \alpha(A_n+2 \alpha I_n)^{-1} \\
      I_n & - 2\alpha (A_n + 2 \alpha I_n)^{-1}
    \end{array}
  \right]
  \left[ \begin{array}{c}p^d \\ p^a \end{array} \right]
  + D_0 \left[ \begin{array}{c} \alpha (A_n (A_n + 2 \alpha I_n))^{-1} E_n \\ (A_n + 2 \alpha I_n)^{-1} E_n \end{array} \right].
    \label{eqn:solution-multicp}
\end{eqnarray}
Let us now verify that the revenues to each of CPs and the ISP are nonnegative. First we handle the ISP. Observe that the components of $p^s + p^d$ constitute revenues from each family of internauts. From (\ref{eqn:solution-multicp}) we gather that
\begin{eqnarray}
  p^s + p^d
   & =& \alpha (A_n + 2 \alpha I_n)^{-1} p^a + \alpha D_0 (A_n + 2 \alpha I_n)^{-1} A_n^{-1} E_n \nonumber \\
  \label{eqn:isp-receipt}
  & = & \alpha A_n^{-1} (A_n + 2 \alpha I_n)^{-1} (A_n p^a + D_0 E_n).
\end{eqnarray}
Next, consider the CPs. Again from (\ref{eqn:solution-multicp}) we gather that
\begin{eqnarray}
  p^c - p^d + p^a & = & (I_n - 2 \alpha (A_n + 2 \alpha I_n)^{-1}) p^a + D_0 (A_n + 2 \alpha I_n)^{-1} E_n \nonumber \\
  \label{eqn:cp-receipt}
  & = & (A_n + 2 \alpha I_n)^{-1} (A_n p^a + D_0 E_n).
\end{eqnarray}
From (\ref{eqn:solution-multicp}), it also follows that
\[
  p^s + p^c = (A_n + 2 \alpha I_n)^{-1} (- \alpha p^a + D_0 (I_n + \alpha A_n^{-1}) E_n)
\]
so that the demand vector $d = D_0 E_n - A_n (p^s + p^c)$ can be written (after observing that all involved matrices commute) as
\begin{equation}
  \label{eqn:demand-solution}
  d = \alpha (A_n + 2 \alpha I_n)^{-1} ( A_n p^a + D_0 E_n).
\end{equation}
Using this in (\ref{eqn:isp-receipt}), we see that $p^s + p^d = A_n^{-1} d$ so that
\begin{equation}
  \label{eqn:isp-utility}
  U_{\isp} = d^T A_n^{-1} d.
\end{equation}

Necessity of
\begin{equation}
  \label{eqn:necessary-multicp-positivedemand}
  (A_n + 2 \alpha I_n)^{-1} (A_n p^a + D_0 E_n) > 0
\end{equation}
is then clear from (\ref{eqn:cp-receipt}) and (\ref{eqn:demand-solution}). Indeed, if any component on the left-hand side of (\ref{eqn:necessary-multicp-positivedemand}) is nonpositive, the corresponding CP derives a nonpositive revenue per unit demand, and the demand for this CP's content truncates to 0. Such a point is either not an equilibrium, or if so, not all demands are strictly positive.

Sufficiency of (\ref{eqn:necessary-multicp-positivedemand}) is obtained as follows. If (\ref{eqn:necessary-multicp-positivedemand}) holds, then (\ref{eqn:cp-receipt}), (\ref{eqn:demand-solution}), and (\ref{eqn:isp-utility}) yield a point with strictly positive revenue for all agents and strictly positive demand. Indeed, from (\ref{eqn:cp-receipt}), revenue per unit demand is strictly positive for all CPs; from (\ref{eqn:demand-solution}), all demands are strictly positive and consequently, all CPs' utilities are strictly positive; from (\ref{eqn:isp-utility}) and the fact that $A_n^{-1}$ has strictly positive eigenvalues, the ISP revenue is also positive. Furthermore, this point satisfies first-order optimality conditions. Given the concavity of the utility functions, it is a Nash equilibrium.

We have thus established that (\ref{eqn:necessary-multicp-positivedemand}) is necessary and sufficient for a pure strategy Nash equilibrium to exist. When this holds, the pure strategy Nash equilibria are such that (\ref{eqn:isp-receipt})-(\ref{eqn:isp-utility}) hold, for a given $p^d$.

Let us now bring the relative weights into the picture. Since the choice of $p^d$ does not affect the demands $d_i(p^s + p^d)$ as in (\ref{eqn:demand-solution}), and the collections per unit demand by each of the CPs and the ISP are as in (\ref{eqn:cp-receipt}) and (\ref{eqn:isp-receipt}), respectively, the optimal solution $p^d$ to the sharing problem can be taken as any vector.

It then follows that the unique demand is given by (\ref{eqn:demand-solution}) which establishes the fourth item. The form of the solution for $p^s$ and $p^c$ in (\ref{eqn:solution-multicp}) shows that $p^s_i = g_i - p^d_i$ and $p^c_i = h_i + p^d_i$ which verifies the second and the third items. Notice that $p^d$ can be any vector, and so the solution is unique up to a free choice of $p^d$, and the statement of the first item is verified. The last item follows from the observation that the demand vector, the price charged by the CPs in (\ref{eqn:cp-receipt}), and the the revenue of the ISP in (\ref{eqn:isp-utility}) do not depend on $p^d$. This concludes the proof.

\section{Proof of Theorem \ref{thm:multicp-bargainafter}}
\label{app:proof-bargain-after}
The system has two CPs, $n = 2$. When $\beta = 0$, the problem separates into two smaller problems each with one CP and one ISP, and Theorem \ref{thm:1cp-bargainafter} applies. We now assume $\alpha > \beta > 0$. It will be useful to recall Figure \ref{fig:Demand Region} which has four regions.

1. We now argue there are no pure strategy equilibria in Region 2. This is the region with both demands zero. Let the ISP prices $p^s_1, p^s_2$ and the CP prices $p^c_1$ and $p^c_2$ be such that \begin{eqnarray*}
  p_1 & = & p^s_1 + p^c_1 \geq D_0/(\alpha - \beta) \\
  p_2 & = & p^s_2 + p^c_2 \geq D_0/(\alpha - \beta).
\end{eqnarray*}
Consider the point O in Figure \ref{fig:Demand Region} given by $(D_0/(\alpha - \beta), D_0/(\alpha - \beta))$. ISP can bring down both his prices to move the price point to O, and demand and revenue collected remain zero. Now consider further deviation along the line BO. To realise this, ISP reduces both prices so that the net price denoted $(q_1, q_2)$ satisfies the equation $q_2 = (D_0 + \beta q_1)/\alpha$. Along this line $d_1 = D'_0 - \alpha' q_1 > 0$ and $d_2 = 0$, where $D'_0$ and $\alpha'$ are given in the statement of the theorem. But this puts us in a single-CP single-ISP case. By the last part of the proof of Theorem \ref{thm:1cp-bargainafter}, we see that the ISP has a deviation that yields a strictly positive revenue for itself. So no point in Region 2 can be a pure strategy equilibrium.

2. We now argue that no point in the interior of Region 1 can be an equilibrium. Let the prices be such that the total prices on the internauts is $(p_1, p_2)$, a point in Region 1. In this case
\begin{eqnarray}
  d_1(p_1, p_2) = D_0 - \alpha p_1 + \beta p_2 > 0 \nonumber \\
  d_2(p_1, p_2) = D_0 - \alpha p_2 + \beta p_1 > 0. \label{eqn:Th6-d1d2-positive}
\end{eqnarray}
Clearly, the net revenue coming from internauts $i$ is $p_i + p^a_i$, and so
\begin{eqnarray*}
  U_{\isp} & = & d_1(p_1, p_2) \gamma_1 (p_1 + p^a_1) + d_2(p_1, p_2) \gamma_2 (p_2 + p^a_2) \\
  U_{\cp, i} & = & d_i(p_1, p_2) (1 - \gamma_i) (p_i + p^a_i), \quad i = 1,2.
\end{eqnarray*}
Since the utilities depend on $p^s_i$ and $p^c_i$ only through $p_i = p^s_i + p^c_i$, partial derivatives with respect to $p^s_i$ and $p^c_i$ may be obtained by considering partial derivatives with respect to $p_i$. These are (in Region 1)
\begin{eqnarray*}
  \frac{\partial U_{\cp,i}}{\partial p_i} & = & (1 - \gamma_i) (d_i(p_1, p_2) -\alpha (p_i + p^a_i)), \quad i=1,2 \\
  \frac{\partial U_{\isp}}{\partial p_1} & = & \gamma_1 (d_1(p_1, p_2) -\alpha (p_1 + p^a_1)) + \gamma_2 \beta (p_2 + p^a_2) \\
      & = & \frac{\gamma_1}{1 - \gamma_1} \frac{\partial U_{\cp_1}}{\partial p_1} + \gamma_2 \beta (p_2 + p^a_2) \\
  \frac{\partial U_{\isp}}{\partial p_2} & = & \gamma_2 (d_2(p_1, p_2) -\alpha (p_2 + p^a_2)) + \gamma_1 \beta (p_1 + p^a_1) \\
     & = & \frac{\gamma_2}{1 - \gamma_2} \frac{\partial U_{\cp_2}}{\partial p_2} + \gamma_1 \beta (p_1 + p^a_1).
\end{eqnarray*}
(In passing, we note that from here, it is but a short step to verify that the Hessian for $U_{\isp}$ with respect to $(p^s_1, p^s_2)$ is $-2 \alpha H$). The first order necessary conditions imply that the above partial derivatives are zero, and we immediately deduce that $p_i + p^a_i = 0$ for both $i=1,2$, i.e., the revenue for each CP's content is zero. Substitution of these in $\frac{\partial U_{\cp_i}}{\partial p_i} = 0$ above yields $d_i = 0$ for both $i = 1, 2$. But this is contrary to the assumption that the point is on the interior of Region 1. So no point in the interior of Region 1 can be an equilibrium.

3. Let us now consider a candidate equilibrium in Region 4, with $p_1 < D_0 / (\alpha - \beta)$ and $p_2 \geq (D_0 + \beta p_1) / \alpha$.

Let us consider deviations by the ISP. First, he may reduce $p^s_2$ to $q^s_2$ so that $p_2$ reduces to $q_2 = (D_0 + \beta p_1) / \alpha$ so that the resulting point $(p_1, q_2)$ is on the line BO, and $d_2$ is on the threshold of positivity, but revenue of CP 1, revenue of CP 2 (which is zero), and revenue of the ISP still remain unaffected. ISP can now consider deviations from $(p_1, q_2)$ along the line BO or into Region 1, i.e., along the vector $(u_1, u_2)$ where $u_2 \leq (\beta / \alpha)u_1$. For such deviations to be fruitless, $\nabla U_{\isp}(p_1,q_2)$ must point into Region 4, and must in particular be normal to the line BO, and so (\ref{eqn:partials-condition}) should hold, which in the present case yields
\begin{eqnarray*}
  \lefteqn{ \gamma_2 d_2(p_1, q_2) - \alpha \gamma_2 (p_2 + p^a_2) + \gamma_1 \beta (p_1 + p^a_1) } \\
  & = & ~ -(\alpha / \beta) (\gamma_1 d_1(p_1, q_2) - \alpha \gamma_1 (p_1 + p^a_1) + \gamma_2 \beta (p_2 + p^a_2)).
\end{eqnarray*}
After cancelations and after using the fact that $d_2(p_1, q_2) = 0$, the above equality simplifies to
\begin{eqnarray}
  p_1 + p^a_1
    & = & \alpha d_1 (p_1, q_2) / (\alpha^2 - \beta^2) \nonumber \\
    & = & d_1 (p_1, q_2) / \alpha' \label{eqn:p1-necessity-bargain-after-1}\\
    & = & (D'_0 - \alpha' p_1) / \alpha', \nonumber
\end{eqnarray}
solving which we get
\begin{eqnarray}
  \label{eqn:p1-necessity-bargain-after}
  p_1 = \frac{D'_0 - \alpha' p^a_1}{2 \alpha'},
\end{eqnarray}
the solution for the single-CP and single-ISP case. It is easily verified that the net revenue is $p_1 + p^a_1 = (D'_0 + \alpha' p^a_1)/(2 \alpha') > 0$ and further, from (\ref{eqn:p1-necessity-bargain-after-1}),
\[
  d_1(p_1, q_2) = \alpha' (p_1 + p^a_1) = (D'_0 + \alpha' p^a_1)/2 > 0,
\]
as in the single-CP and single-ISP case.

Let us next consider local deviations by CP 1 who can increase or decrease $p^c_1$ and therefore perturb $p_1$. From the above argument, $p_1$ must satisfy (\ref{eqn:p1-necessity-bargain-after}). If $(p_1, p_2)$ is an interior point of Region 4, any deviation by CP 1 effectively moves the point $(p_1, q_2)$, a point that is effectively equivalent to the original $(p_1, p_2)$, along the line BO. It is easy to see, using $q_2 = (D_0 + \beta p_1)/\alpha$, that
\begin{eqnarray}
  \frac{\partial}{\partial p_1} U_{\cp,1} & = & \frac{\partial}{\partial p_1}(d_1(p_1,q_2) (p_1 + p^a_1) (1-\gamma_1)) \nonumber \\
  & = & \frac{\partial}{\partial p_1}((D'_0 - \alpha' p_1)  (p_1 + p^a_1) (1-\gamma_1)) \nonumber \\
  & = & (d_1(p_1, q_2) - \alpha' (p_1 + p^a_1)) (1-\gamma_1) \nonumber \\
  & = & 0 \label{eqn:cp-1-deviation-bargain-after-1}
\end{eqnarray}
where the last equality comes from (\ref{eqn:p1-necessity-bargain-after-1}). Thus, when $(p_1, p_2)$ is an interior point of Region 4, CP 1 does not benefit from a local deviation. When $(p_1, p_2)$ is on the line BO, it is just $(p_1, q_2)$. A decrease in $p_1$ moves the point to the interior of Region 4, and the equivalent point moves lower and left along the line BO. Then the argument leading to (\ref{eqn:cp-1-deviation-bargain-after-1}) holds for the left partial derivative $\frac{\partial^- }{\partial p_1} U_{\cp,1}$, and decrease in $p^c_1$ does not yield a gain. On the other hand, an increase in $p^c_1$ increases $p_1$ and puts the system in the interior of Region 1, and we then have
\begin{eqnarray*}
  \frac{\partial^+ }{\partial p_1} U_{\cp,1} & = & \frac{\partial}{\partial p_1}(d_1(p_1,q_2) (p_1 + p^a_1) (1-\gamma_1)) \nonumber \\
  & = & \frac{\partial}{\partial p_1}((D_0 - \alpha p_1 + \beta q_2) (p_1 + p^a_1) (1-\gamma_1)) \nonumber \\
  & = & (d_1(p_1, q_2) - \alpha (p_1 + p^a_1)) (1-\gamma_1) \nonumber \\
  & = & (d_1(p_1, q_2) - \alpha' (p_1 + p^a_1)) (1-\gamma_1) + (\alpha' - \alpha) (p_1 + p^a_1) (1 - \gamma_1) \nonumber \\
  & = & (\alpha' - \alpha)(p_1 + p^a_1)(1 - \gamma_1) \\
  & < & 0
\end{eqnarray*}
where the penultimate equality follows because of (\ref{eqn:p1-necessity-bargain-after-1}), and the last inequality follows because $\alpha' < \alpha$, but the other two factors are strictly positive. But this implies an infinitesimal increase in $p^c_1$ yields a strict decrease in his utility. There are thus no utility increasing infinitesimal deviations for CP 1.

Lastly, we consider infinitesimal deviations by CP 2. If $p_2 > q_2$, then CP 2 can bring down his price so that $p_2$ reduces to $q_2$ without a change in his revenue or without a change in demand for his content. Any further decrease moves the operating point into the interior of Region 1, and renders $d_2$ strictly positive. For such a deviation to be fruitless, the revenue for CP 2 at the operating point $(p_1, q_2)$ should be nonpositive, i.e.,
\begin{eqnarray*}
  0 \geq q_2 + p^a_2 = \frac{D_0 + \beta p_1}{\alpha} + p^a_2.
\end{eqnarray*}
Substitution of (\ref{eqn:p1-necessity-bargain-after}) and rearrangement yields (\ref{eqn:too-high-pa}) as a necessary condition for equilibrium. If (\ref{eqn:too-high-pa}) does not hold, there is no pure strategy Nash equilibrium in Region 4.

4. Consider points in Region 3. An argument analogous to above yields that an analogue of (\ref{eqn:too-high-pa}), with indices 1 and 2 interchanged, is a necessary condition. But as $p^a_1 \geq p^a_2$, such a condition cannot hold, and there is no pure strategy Nash equilibrium in Region 3

It is thus clear that if (\ref{eqn:too-high-pa}) does not hold, there exists no pure strategy Nash equilibrium. This proves the second statement. If (\ref{eqn:too-high-pa}) does hold, we saw above that the only possible equilibria, if any, are in Region 4 with $p_1$ as in (\ref{eqn:p1-necessity-bargain-after}), and $p_2 \geq (D_0 + \beta p_1) / \alpha$. From the first order conditions, no infinitesimal deviation yields a better revenue for any of the agents. From the facts that
\begin{itemize}
  \item $U_{\isp}$ is concave in Region 1 by the assumption that $H$ is positive definite,
  \item $U_{\cp,1}$ and $U_{\cp,2}$ are strictly concave in Region 1,
  \item they extend continuously to the boundaries AO and BO from Region 1,
  \item for each point in Region 4, the utilities are determined by the utilities on an equivalent point on the line BO, and similarly,
  \item for each point in Region 3, the utilities are determined by the utilities on an equivalent point on the line AO, and finally,
  \item the utilities earned in Region 2 are zero,
\end{itemize}
it follows that no deviation, infinitesimal or otherwise, will yield a better revenue for any of the agents. So $p_1$ given in (\ref{eqn:p1-necessity-bargain-after}) and $p_2 \geq (D_0 + \beta p_1) / \alpha$ characterize the pure strategy Nash equilibrium. This concludes the proof of the theorem.

\section{Proof of Theorem \ref{thm:NEExclusiveContract}}
\label{app:proof-nonneutrality-collusion}
Recall our argument that $U_{\ispbar}$ is a concave function in $\overline{p}^s=(p_1, p_2^s, \cdots, p_n^s)$, and for each $k=2,\cdots, n, U_{CP, k}$ is concave in $p_k^c$. Then, the equilibrium prices must satisfy the following first order optimality conditions:
\begin{eqnarray*}
\frac{\partial U_{\overline{ISP}}}{\partial p_1}&=&D_0-2\alpha p_1 + 2 \beta \sum_{j\neq 1}p_j^s + \beta \sum_{j\neq 1}p_j^c - \alpha p_1^a + \beta \sum_{j \neq 1} p_i^d=0,
\end{eqnarray*}
and for $k=2,3,\cdots n$,
\begin{eqnarray*}
\frac{\partial U_{\overline{ISP}}}{\partial p_k^s}
& = &  D_0 + 2 \beta p_1-2 \alpha p_k^s + 2 \beta \sum_{j\neq k,1}p_j^s-\alpha p_k^c + \beta \sum_{j\neq k,1}p_j^c + \beta p_1^a -\alpha p_k^d + \beta \sum_{j\neq k,1}p_j^d=0 \\
\frac{\partial U_{CP,k}}{\partial p_k^c}
& = & D_0 + \beta p_1 -\alpha p_k^s  + \beta \sum_{j\neq k,1}(p_j^s+p_j^c) - 2\alpha p_k^c-\alpha (p_k^a-p_k^d)=0.
\end{eqnarray*}
Let $\overline{p}^c=(p_2^c,p_3^c,\cdots, p_n^c)$ denote the CP price vector. The above set of $2n-1$ equations can be compactly written in matrix form as follows
\begin{equation}
\label{eqn:CollusionMatrixEquations}
\left[ {\begin{array}{cc}
 2\alpha & -b^T  \\
 -b & C  \\
 \end{array} } \right]
\left[ {\begin{array}{cc}
 \overline{p}^s  \\
 \overline{p}^c  \\
 \end{array} } \right]=
\left[ {\begin{array}{cc}
 -\alpha &  c^T \\
 a & D  \\
 \end{array} } \right]
\left[ {\begin{array}{cc}
p^a \\
\overline{p}^d \\
\end{array} } \right] + D_0 E_{2n-1},
\end{equation}
where $b^T$ and $c^T$ are row vectors of size $1\times (2n-2)$ given by $b^T=[2\beta E_{n-1}^T \; \beta E_{n-1}^T]$ and $c^T=[0 \cdot E_{n-1}^T \; \beta E_{n-1}^T]$. $a$ denotes a column vector of size $(2n-2)\times 1$ given by $a^T=[\beta \cdot E_{n-1}^T \; 0 \cdot E_{n-1}^T]$. $C$ and $D$ are $2\times 2$ block matrix given by
\[
C=\left[ {\begin{array}{cc}
 2A_{n-1} & A_{n-1}  \\
 A_{n-1} & 2B_{n-1}  \\
 \end{array} } \right]\;\;\;\;\;
 D=\left[ {\begin{array}{cc}
 \bigcirc & -A_{n-1}  \\
 -\alpha I_{n-1} & \alpha I_{n-1}  \\
 \end{array} } \right].
 \]
The solution to system of equation in (\ref{eqn:CollusionMatrixEquations}) exists if the matrix    \[
 \left[ {\begin{array}{cc}
 2\alpha &  -b^T\\
 -b & C  \\
 \end{array} } \right]
 \]
is invertible. By inspection we can write its inverse as
\begin{equation}
\label{eqn:CollusionInverseMatrix}
\frac{1}{\mu}
\left[ {\begin{array}{cc}
1 &  b^TC^{-1} \\
C^{-1}b & \mu C^{-1}+C^{-1}bb^TC^{-1}  \\
 \end{array} } \right],
 \end{equation}
where $\mu= (2\alpha-b^TC^{-1}b)$ and $C^{-1}$ denotes the inverse of matrix $C$. For the above inverse to exist
the following must hold.
\begin{itemize}
\item[(i)] $C$ is invertible and
\item[(ii)] $b^TC^{-1}b\neq 2\alpha $.
\end{itemize}
We next verify these conditions. Invertibility of $C$ is guaranteed by its definition. Indeed,
\begin{eqnarray*}
\text{det}(C)&=&\text{det}(A_{n-1}(A_{n-1}+2\alpha I_{n-1}))\\
&=&(\alpha+\beta)^{(n-2)}(\alpha-\beta(n-2))(3\alpha+\beta)^{(n-2)} (3\alpha-\beta(n-2)) > 0,
\end{eqnarray*}
and it can be computed as
\begin{eqnarray*}
C^{-1}=
 (A_{n-1}+2\alpha I_{n-1})^{-1}\circ\left[ {\begin{array}{cc}
 2B_{n-1}A_{n-1}^{-1} &
 -I_{n-1} \\
 -I_{n-1} & 2I_{n-1}  \\
 \end{array} } \right].
 \end{eqnarray*}
Further, all the terms in matrix  ($\ref{eqn:CollusionInverseMatrix}$) can be expressed in term of inverse of matrix $A_{n-1}$ as follows:

\begin{eqnarray}
\label{eqn:collusionInversermatrixTerms}
C^{-1}b&=&\beta
\left[ {\begin{array}{cc}
 A_{n-1}^{-1}E_{n-1} \\
 0  \\
 \end{array} } \right], \\
 b^TC^{-1}&=&\beta
\left[ {\begin{array}{cc}
 E_{n-1}^TA_{n-1}^{-1} & 0\\
 \end{array} } \right], \nonumber \\
C^{-1}bb^TC^{-1}&=& \left[ {\begin{array}{cc}
  \beta^2 A_{n-1}^{-1}E_{n-1}E_{n-1}^TA_{n-1}^{-1} &
 \bigcirc \\
 \bigcirc & \bigcirc  \\
 \end{array} } \right]. \nonumber
 \end{eqnarray}
 \noindent
Left multiplying matrix (\ref{eqn:collusionInversermatrixTerms}) by $b^T$ we have
 \begin{eqnarray}
 \label{eqn:CollusionInverseCondition2}
 b^TC^{-1}b&=&2\beta^2 E_{n-1}^TA_{n-1}^{-1}E_{n-1}
 = \frac{2\beta^2 (n-1)}{\alpha -(n-2)\beta}.
\end{eqnarray}
The above relation follows by noting that the sum of elements in each row of the adjacent matrix of $A_{n-1}$ and its determinant are given by $(\alpha+\beta)^{n-2}$ and $(\alpha+\beta)^{n-2}(\alpha-(n-2)\beta)$, respectively.  The left hand side in (\ref{eqn:CollusionInverseCondition2}) is equal to $2\alpha$ only when $\alpha=\beta$ and $n=2$. But, this contradicts our assumption $\alpha> (n-1)\beta $. This completes the proof of existence of equilibrium.

We next compute the equilibrium prices and the corresponding demand. Rearranging (\ref{eqn:CollusionMatrixEquations}), equilibrium prices can be written as
 \begin{eqnarray}
 \nonumber
\left[ {\begin{array}{cc}
 \overline{p}^s  \\
 \overline{p}^c  \\
 \end{array} } \right]&=&
 \frac{1}\mu
\left[ {\begin{array}{cc}
1 &  b^TC^{-1} \\
C^{-1}b & Y  \\
 \end{array} } \right]
\left[ {\begin{array}{cc}
 -\alpha &  c^T \\
 a & D  \\
 \end{array} } \right]
\left[ {\begin{array}{cc}
p^a \\
\overline{p}^d \\
\end{array} } \right]
+ \frac{D_0 }\mu
\left[ {\begin{array}{cc}
1 &  b^TC^{-1} \\
C^{-1}b & Y  \\
 \end{array} } \right] E_{2n-1},
\end{eqnarray}
where $Y=\mu C^{-1}+C^{-1}bb^TC^{-1}$.

To further simplify the expression for equilibrium prices, we use the relations $2b^TC^{-1}a=b^TC^{-1}b$ and $c^T+b^TC^{-1}D=0$, both of which are easy to verify. Using these relations we get
\begin{eqnarray}
 \nonumber
 \left[ {\begin{array}{cc}
 \overline{p}^s  \\
 \overline{p}^c  \\
 \end{array} } \right]&=&
 \left[ {\begin{array}{cc}
-1/2 &  \bigcirc \\
C^{-1}(a-b/2) & C^{-1}D  \\
 \end{array} } \right]
\left[ {\begin{array}{cc}
p^a \\
\overline{p}^d \\
\end{array} } \right] +
 \frac{D_0 }{\mu}
\left[ {\begin{array}{cc}
1 &  b^TC^{-1} \\
C^{-1}b & Y  \\
 \end{array} } \right] E_{2n-1} \nonumber\\
&=&
 \left[ {\begin{array}{cc}
-1/2 &  \bigcirc \\
C^{-1}(a-b/2) & C^{-1}D  \\
 \end{array} } \right]
\left[ {\begin{array}{cc}
p^a \\
\overline{p}^d \\
\end{array} } \right] +
 D_0
\left[ {\begin{array}{cc}
-1/2\beta +  (\beta+\alpha)/\beta\mu \\
C^{-1}b(\frac{-1}{2\beta}+\frac{\beta+\alpha}{\beta \mu})+C^{-1}E_{2n-2}
  \\
 \end{array} } \right] \nonumber\\
\label{eqn:CollusionEquilibriumPriceDependency}
 &=&
 \left[ {\begin{array}{cc}
-1/2 &  \bigcirc \\
C^{-1}(a-b/2) & C^{-1}D  \\
 \end{array} } \right]
\left[ {\begin{array}{cc}
p^a \\
\overline{p}^d \\
\end{array} } \right]
+
 D_0
\left[ {\begin{array}{cc}
-1/2\beta +  (\beta+\alpha)/\beta\mu \\
\frac{X^{-1}}{2}\circ\left[ {\begin{array}{cc}
-E \\
2E \\
\end{array} } \right]+\frac{\beta+\alpha}
{\mu}\left[ {\begin{array}{cc}
A^{-1}E \\
\bigcirc \\
\end{array} } \right]
  \\
 \end{array} } \right],
 \label{eqn:CollusionEquilibrium}
\end{eqnarray}
where $X=(A_{n-1}+2\alpha I_{n-1})^{-1}$, and we used notation $A:=A_{n-1}$ and $E:=E_{n-1}$ for ease of presentation.

The product of matrix $C^{-1}$ and $D$ in the above expression can be computed as
\begin{equation}
C^{-1}D=\left[
  \begin{array}{cc}
    \alpha X^{-1} & I_{n-1} \\
    -2\alpha X^{-1} & -I_{n-1} \\
  \end{array}
\right]
\end{equation}
Substituting this relation in (\ref{eqn:CollusionEquilibriumPriceDependency}), it is easy to see that equilibrium prices $p_i^s$ and $p_i^c$ depend only on $p_i^d$ all $i=2,3,\cdots,n$. This verifies the claims in second item. Further, using (\ref{eqn:CollusionEquilibriumPriceDependency}), it follows that at equilibrium prices paid by each user group can be computed as
\begin{eqnarray*}
 \left[ {\begin{array}{cc}
 p_1  \\
 \overline{p}^s_{-1}+\overline{p}^c  \\
 \end{array} } \right] &=&
\left[ {\begin{array}{cc}
-1/2 &  \bigcirc \\
-(\beta/2)X^{-1}E & -\alpha X^{-1} \\
 \end{array} } \right]
p^a +
 D_0
\left[ {\begin{array}{cc}
-1/2\beta +  (\beta+\alpha)/\beta \mu \\
X^{-1}E/2+(\beta+\alpha)
/\mu A^{-1}E
  \\
 \end{array} } \right],
\end{eqnarray*}
where $\overline{p}^s_{-1}$ denotes the ISP price vector without the component $p_1$. The corresponding equilibrium demand can be computed as
\begin{eqnarray*}
D_0E_n-\left[
  \begin{array}{c}
    -\frac{\alpha}{2}+ \frac{\beta^2}{2E^TXE} +\beta\alpha E^TX \\
    \frac{\beta}{2} E-\frac{\beta}{2} AE-\alpha AX \\
  \end{array}
\right] -
D_0\left[
  \begin{array}{c}
    \frac{\alpha}{2(\alpha-(n-1)\beta)} -\beta E^T XE-\frac{\beta(\beta+\alpha)}{\mu} E^T A^{-1}E\\
    \frac{-\beta}{2(\alpha-(n-1)\beta)}E + \frac{(\beta+\alpha)}{\mu}E+AXE/2\\
  \end{array}
\right].
\label{eqn:CollusionPositveDemandEquiCondn}
\end{eqnarray*}
Note that both equilibrium prices and the corresponding demand do not depend on $p_d$, verifying the claim in the last item. Also, notice that $p_d$ can be any vector, and so the solution is unique up to a free choice of $p_d$, and the statement of the first item is verified.

\section{Dynamics}
\label{app:Dynamics}
In this section we constrain prices to remain in Region $1$ of Figure \ref{fig:Demand Region}.
This yields a coupled constraint which is a significant difference with
respect to the unconstrained model in Section \ref{sec:multiple-CP}. For this new setting,
we discuss two dynamic models with multiple content providers. Again,
for ease of exposition, we restrict to the case of two CPs.

\subsection{Continuous dynamics}

Let us assume that the players set their prices such that the demand from each CP is nonnegative, i.e., (\ref{eqn:demand-multi-CP}) is greater than or equal to zero for both CPs. This imposes coupled constraints on the set of prices $(p^s,p^c) \in \mathbb{R}^4$ given by
\begin{eqnarray*}
& d_1(p^s,p^c)\geq 0, & d_2(p^s,p^c)\geq 0,\\
& p_1^s+p_1^d \geq 0, & p_2^s+p_2^d \geq 0, \\
& p_1^c-p_1^d+p_1^a \geq 0, & p_2^c-p_2^d+p_2^a \geq 0.
\end{eqnarray*}
Let $R$ denote the set of prices that satisfy these constraints. It is easy to verify that the above constraints also result in the following upper bounds on the prices:
\begin{eqnarray*}
& p_1^s \leq \frac{D_0}{\alpha-\beta}+p_1^a -p_1^d, & p_2^s \leq \frac{D_0}{\alpha-\beta}+p_2^a -p_2^d\\
& p_1^c \leq \frac{D_0}{\alpha-\beta}+ p_1^d, & p_2^c \leq \frac{D_0}{\alpha-\beta}+ p_2^d,
\end{eqnarray*}
and thus the set $R$ is compact. Furthermore, due to the linearity of the constraints in the prices, $R$ is convex. As argued in Section \ref{sec:multiple-CP}, for any price vector $p=(p^s,p^c) \in R$, the mappings $U_{\isp}(\cdot,p^c), U_{\cp,1}(p^s,\cdot,p_2^s) \text{\;and\;} U_{\cp,1}(p^s,p_1^c,\cdot)$ are concave functions in the ``$\cdot$'' variables.

Given the concave utility functions defined on the coupled constraint set $R$, we are in the setting of $n$-person concave games studied by Rosen \cite{rosen}. We can then directly use the dynamic model proposed by Rosen \cite[Sec. 4]{rosen}. In our game setting the system of differential equations for the strategies $p_1^s,p_2^s,p_1^c, p_2^c$ is:
\[\frac{{\rm d} p^s_1}{\rm d t}=\frac{\partial{U_{\isp}(p^s,p^c)}}{\partial p^s_1}+u_1(p)\frac{\partial d_1(p)}{\partial p^s_1}+u_2(p)\frac{\partial d_2(p)}{\partial p^s_1} + u_3(p)\]
\[\frac{{\rm d} p^s_2}{\rm d t}=\frac{\partial{U_{\isp}(p^s,p^c)}}{\partial p^s_2}+u_1(p)\frac{\partial d_1(p)}{\partial p^s_2}+u_2(p)\frac{\partial d_2(p)}{\partial p^s_2} + u_4(p)\]
\[\frac{{\rm d} p^c_1}{\rm d t}=\frac{\partial{U_{\cp,1}(p^s,p^c)}}{\partial p_1^c}+u_1(p)\frac{\partial d_1(p)}{\partial p_1^c}+u_2(p)\frac{\partial d_2(p)}{\partial p_1^c} + u_5(p)\]
\[\frac{{\rm d} p^c_2}{\rm d t}=\frac{\partial{U_{\cp,2}(p^s,p^c)}}{\partial p_2^c}+u_1(p)\frac{\partial d_1(p)}{\partial p_2^c}+u_2(p)\frac{\partial d_2(p)}{\partial p_2^c} + u_6(p).\]
In the above dynamics it is assumed that a central agent computes $u(p)=(u_1(p),\ldots, u_6(p))$ as in \cite[eqn.~(4.5)]{rosen} and communicates the values to the players. The above dynamics tend to an equilibrium as is established next.

\begin{thm}
Let $\alpha>\beta$. Starting from any point $p\in R$, the continuous solution $p(t)$ to the above system of differential equations remains in $R$ for all $t$ and converges to the unique equilibrium point.
\end{thm}
\begin{proof}
The first claim follows directly from Rosen's \cite[Th. 7]{rosen}. To prove the second part, we verify the so-called {\em diagonal strict concavity} property of
\[
  \sigma(p)=U_{\isp}(p)+U_{\cp,1}(p)+U_{\cp,2}(p), \quad p \in R.
\]
Let $g(p)$ denote the gradient of $\sigma(p)$ given by
\[g(p)=\left[
           \begin{array}{c}
              \partial U_{\isp}(p)/\partial p_1^s\\
              \partial U_{\isp}(p)/\partial p_2^s \\
              \partial U_{\cp,1}(p)/\partial p_1^c\\
              \partial U_{\cp,2}(p)/\partial p_2^c\\
           \end{array}
         \right].
\]
With $\tau := \beta/\alpha$, the Jacobian $G(p)$ of the above matrix can be verified to be the symmetric matrix
\[G(p)=-\alpha \left[
           \begin{array}{cccc}
             2 & -2\tau & 1 & -\tau \\
             -2\tau & 2 & -\tau & 1 \\
             1 & -\tau & 2 & -\tau \\
             -\tau & 1 & -\tau & 2 \\
           \end{array}
         \right].
\]
It is easy to see that the eigenvalues of $-G(p)/\alpha$ are
\begin{eqnarray*}
  & ((3\tau+4)\pm\sqrt{(3\tau+4)^2-4(\tau^2+4\tau+3)})/2 & \\
  & ((4-3\tau)\pm \sqrt{(4-3\tau)^2-4(\tau^2-4\tau+3)})/2, &
\end{eqnarray*}
and that these eigenvalues are strictly positive for $\tau \in [0, 1)$. $G(p)$ is therefore negative definite. By \cite[Th.~6]{rosen}, $\sigma(p)$ is diagonally strictly concave, and by \cite[Th.~9]{rosen}, the equilibrium point is globally asymptotically stable for the system of differential equations; this establishes convergence.
\end{proof}

\subsection{Discrete dynamics}
In this subsection we study discrete dynamics motivated by the best response dynamics. We assume the providers set their price, say, at the beginning of each day, as the best response to prices set by the other players on the previous day.

Let $p_t=((p_{1t}^s,p_{2t}^s), p_{1t}^c,p_{2t}^c)$ denote the price set by the players on day $t$. Recalling the concavity properties of the utility functions, the price set by the players on day $m=t+1$ are obtained by setting
\begin{align}
\label{eqn:DynamicsBestResponseEqnsISP1}
\partial U_{\isp}((p_{1m}^s,p_{2m}^s), p_{1t}^c,p_{2t}^c)/\partial p_{1}^s&=0 \\
\label{eqn:DynamicsBestResponseEqnsISP2}
\partial U_{\isp}((p_{1m}^s,p_{2m}^s), p_{1t}^c,p_{2t}^c)/\partial p_{2}^s &=0 \\
\label{eqn:DynamicsBestResponseEqnsCP1}
\partial U_{\cp,1}((p_{1t}^s,p_{2t}^s), p_{1m}^c,p_{2t}^c)/\partial p_{1}^c &=0 \\
\label{eqn:DynamicsBestResponseEqnsCP2}
\partial U_{\cp,2}((p_{1t}^s,p_{2t}^s), p_{1t}^c,p_{2m}^c)/\partial p_{2}^s &=0.
\end{align}
\noindent
The ISP controls the price $(p_{1m}^s,p_{2m}^s)$ and sets them so that both (\ref{eqn:DynamicsBestResponseEqnsISP1}) and (\ref{eqn:DynamicsBestResponseEqnsISP2}) are simultaneously satisfied. The above conditions straightforwardly result in the following equations:
\begin{eqnarray*}
2\alpha p_{1m}^s-2\beta p_{2m}^s=D_0-\alpha p_{1t}^c + \beta p_{2t}^c -\alpha p_1^d + \beta p_2^d\\
2\alpha p_{2m}^s-2\beta p_{1m}^s=D_0-\alpha p_{2t}^c + \beta p_{1t}^c -\alpha p_2^d + \beta p_1^d\\
2\alpha p_{1m}^c=D_0-\alpha p_{1t}^s + \beta p_{2t}^s + \beta p_{2t}^c -\alpha (p_1^a -p_1^d)\\
2\alpha p_{2m}^c=D_0-\alpha p_{2t}^s + \beta p_{1t}^s + \beta p_{1t}^c -\alpha (p_2^a -p_2^d).
\end{eqnarray*}
This is a linear mapping that can be compactly written in the matrix form as
\begin{equation}
\label{eqn:DynamicsLinearMap}
p_{t+1}^T = X p_t^T + Y
\end{equation}
where, with $\tau = \beta/\alpha$, we take
\[
  X=\hspace{-.1cm}\frac{1}{2}\left[
          \begin{array}{cccc}
            0 & 0 & -1 & 0 \\
            0 & 0 & 0 & -1 \\
            -1 & \tau & 0 & \tau \\
            \tau & -1 & \tau & 0 \\
          \end{array}
        \right],
  Y=\hspace{-.1cm}\left[
  \begin{array}{c}
    \frac{D_0-(\alpha-\beta)p_1^d}{2(\alpha-\beta)} \\
    \frac{D_0-(\alpha-\beta)p_2^d}{2(\alpha-\beta)} \\
    \frac{D_0-\alpha(p_1^a-p_1^d)}{2\alpha} \\
    \frac{D_0-\alpha(p_2^a-p_2^d)}{2\alpha} \\
  \end{array}
  \right].
\]
An easy guess of the fixed point to the iteration in (\ref{eqn:DynamicsLinearMap}) is
\begin{equation}
  \label{eqn:p-opt}
  p_{\opt}^T = (I - X)^{-1} Y.
\end{equation}
Under the assumptions of Theorem \ref{thm:multicp-bargainbefore} for two CPs ($n=2$), it can be verified that $p_{\opt}^T$ is the solution of that theorem given in (\ref{eqn:solution-multicp}). Under the same assumptions, the dynamics converge to that solution, as guaranteed next.


\begin{thm}
For $\tau \in [0, 1)$, the dynamics given in (\ref{eqn:DynamicsLinearMap}) converges to the fixed point $p^T_{\opt} = (I - X)^{-1} Y$.
\end{thm}
\begin{proof}
The eigenvalues of the matrix $X$ can be straightforwardly evaluated; they are
\[
  \frac{\frac{\tau}{2}\pm\sqrt{(\frac{\tau}{2})^2+1-\tau}}{2} \text{\;\;and \;\;}\frac{-\frac{\tau}{2}\pm\sqrt{(\frac{\tau}{2})^2-1-\tau}}{2}.
\]
For $\tau \in [0, 1)$, these eigenvalues are nonzero, of magnitudes strictly smaller than 1, distinct, and hence $X$ is diagonalizable in the form $X=UDU^{-1}$, where $D$ is the diagonal matrix of eigenvalues and $U$ is an invertible matrix. $X$ is also invertible. Consequently, $p^T_{\opt}$ in (\ref{eqn:p-opt}) is well-defined and satisfies
\[
  p^T_{\opt} = X p^T_{\opt} + Y.
\]
Using this and (\ref{eqn:DynamicsLinearMap}), with $p^T_0$ as the initial iterate, the norm of the error at iteration $t+1$ telescopes as
\begin{eqnarray*}
  || p^T_{t+1} - p^T_{\opt} || & = & || X ( p^T_t - p^T_{\opt} ) || \\
  & = & || X^{t+1} (p^T_0 - p^T_{\opt}) || \\
  & = & || (U D U^{-1})^{t+1} (p^T_0 - p^T_{\opt}) || \\
  & = & || U D^{t+1} U^{-1} (p^T_0 - p^T_{\opt}) ||.
\end{eqnarray*}
Since the magnitudes of the eigenvalues are strictly less than 1, the error vector converges to 0 exponentially quickly in the number of iterations.
\end{proof}

{\em Remarks}: 1) The iterates converge if $\tau = \beta / \alpha < 1$. However, to guarantee that the solution is in $R$, we need the other necessary and sufficient condition $(A + 2\alpha I_n)^{-1} [D_0 E_n + A p^a]$ to be made of strictly positive entries.

2) Even if these hold, the iterates may not remain in $R$ due to the coupled nature of the constraints. Strictly speaking then, the dynamics is not the best response dynamics. Indeed, with $D_0=200, \alpha=6, \beta=3, p^d_1=10, p^d_2=25, p^a_1=45, p^a_2=10$, it can be see that with $p_0=(19,2,25,28)$ when demand for both contents is positive, we get $p_1=(15.8333,6.8333,-2.8333,34.1667)$ where demand for CP 1 content alone is positive. Nevertheless, the iterates converge to the unique equilibrium with strictly positive demands.

\end{document}